\documentclass[journal,onecolumn,12pt]{IEEEtran}
\usepackage[table,xcdraw]{xcolor}
\usepackage[numbers,sort&compress]{natbib}
\usepackage{tikz}
\usepackage{adjustbox}
\usepackage{tikz-cd}
\usetikzlibrary{matrix, positioning}
\usepackage{pslatex}
\usepackage{amsfonts,color,morefloats}
\usepackage{amssymb,amsmath,latexsym,amsthm}
\usepackage{hyperref}
\hypersetup{hidelinks}
\usepackage{enumitem}
\usepackage[justification=centering]{caption}

\newcommand{\wt}{{\mathrm{wt}}}

\newcommand{\gf}{{\mathbb{F}}}

\newcommand{\C}{{\mathcal{C}}}

\newtheorem{theorem}{Theorem}
\newtheorem{lemma}[theorem]{Lemma}
\newtheorem{proposition}[theorem]{Proposition}
\newtheorem{corollary}[theorem]{Corollary}

\newtheorem{problem}{Open Problem}

\newtheorem{example}{Example}

\newtheorem{remark}{Remark}

\begin{document}
	\title{On the Hamming Weight Functions of Linear Codes}
\author{\IEEEauthorblockN{	Dongmei Huang\IEEEauthorrefmark{1}, Qunying Liao\IEEEauthorrefmark{1}, Sihem Mesnager\IEEEauthorrefmark{2}\IEEEauthorrefmark{3}, Gaohua Tang\IEEEauthorrefmark{4} 
and Haode Yan\IEEEauthorrefmark{5}
}  \\
	\IEEEauthorblockA{\IEEEauthorrefmark{1}College of Mathematical Science, Sichuan Normal University, Chengdu, China. }\\
	\IEEEauthorblockA{\IEEEauthorrefmark{2}Department of Mathematics, University of Paris VIII, F-93526 Saint-Denis, Laboratory Geometry,
Analysis and Applications, LAGA, University Sorbonne Paris Nord, CNRS, UMR 7539, F-93430,
Villetaneuse, France.}\\
\IEEEauthorblockA{\IEEEauthorrefmark{3}Telecom Paris, Polytechnic Institute of Paris, 91120 Palaiseau, France.}\\
	\IEEEauthorblockA{\IEEEauthorrefmark{4}School of Science, Beibu Gulf University, Qinzhou, China.}\\
	\IEEEauthorblockA{\IEEEauthorrefmark{5}School of Science, Harbin Institute of Technology, Shenzhen, 518055, China}\\
	\IEEEauthorblockA{\href{mailto: qunyingliao@sicnu.edu.cn}{qunyingliao}@sicnu.edu.cn, 
	\href{mailto:B20230801005@stu.sicnu.edu.cn}{huangdongmeimaths}@stu.sicnu.edu.cn,
	\href{mailto:mesnager@univ-paris8.fr}{mesnager}@univ-paris8.fr, 
	\href{mailto:tanggaohua@163.com}{tanggaohua}@163.com, 
	\href{mailto: yanhd@hitZ.edu.cn}{yanhd}@hit.edu.cn\\
	\IEEEauthorblockA{Corresponding Author: Qunying Liao \quad Email: qunyingliao@sicnu.edu.cn}
		}
	
	   


	}
	
		\maketitle
	\begin{abstract}
Currently known secondary construction techniques for linear codes mainly include puncturing, shortening, and extending. In this paper, we propose a novel method for the secondary construction of linear codes based on their weight functions. Specifically, we develop a general framework that constructs new linear codes from the set of codewords in a given code having a fixed Hamming weight. We analyze the dimension, number of weights, and weight distribution of the constructed codes, and establish connections with the extendability of the original codes as well as the partial weight distribution of the derived codes. As a new tool, this framework enables us to establish an upper bound on the minimum weight of two-weight codes and to characterize all two-weight codes attaining this bound. Moreover, several divisibility properties concerning the parameters of two-weight codes are derived. The proposed method not only generates new families of linear codes but also provides a powerful approach for exploring the intrinsic combinatorial and geometric structures of existing codes.
\end{abstract}

	
	
	\section{Introduction}
	
	Let $\mathbb{F}_q$ be the finite field with $q$ elements. A $k$-dimensional subspace $\C$ of the vector space $\mathbb{F}_q^n$, consisting of all $n$-tuples over $\mathbb{F}_q$, is referred to as a \emph{linear code} of length $n$ and dimension $k$. Algebraically, $\C$ is simply a $k$-dimensional vector space over $\mathbb{F}_q$. However, as a specific subspace of $\mathbb{F}_q^n$, $\C$ inherits certain metric properties. In particular, for every $\mathbf{v} \in \mathbb{F}_q^n$, the \emph{weight} of $\mathbf{v}$, denoted by $\mathrm{wt}(\mathbf{v})$, is defined as the number of non-zero entries in the vector $\mathbf{v}$. 	The \emph{distance} between two vectors is defined as the weight of their difference. The interaction between the algebraic structure of $\C$ and the metric structure induced by the weight function is central to the study of coding theory.

	It was demonstrated in \cite{Nogin} that the weight function of a linear code, defined as an integer-valued function on the vector space of messages, uniquely determines the code up to equivalence. Consequently, the weight function plays a crucial role in coding theory, serving as a fundamental tool in distinguishing and characterizing linear codes. 
Hoffman \cite{Hoffman} investigated the conditions under which a weight function defined on an abstract vector space over a finite field can be realized as the weight function of a linear code. In \cite{GHS,Sim}, the degree of the preimage set of the Hamming weight function for linear codes was carefully analyzed. Certain preimage sets with low degrees were demonstrated. By utilizing relevant results concerning the weight distribution of Reed-Muller codes, divisibility properties and lower bounds for the sizes of specific preimage sets of the Hamming weight function were established in these works.

For a linear code, determining the minimum of the Hamming weight function and its value distribution are fundamental but challenging problems, as they correspond to the minimum distance and the weight enumerator of the code. Explicit determination of these parameters typically requires sophisticated tools, including exponential-sum evaluations, Walsh spectra of Boolean functions, and cyclotomic-field techniques.
Yuan, Carlet, and Ding~\cite{YuanCarletDing2006} determined the weight distribution of a class of linear codes derived from perfect nonlinear (planar) functions. Ding~\cite{Ding2015} constructed linear codes from certain $2$-designs using a defining-set approach and characterized their weight distributions via Gauss periods. Ding and Ding~\cite{KDingCDing2015} determined the weight distribution of two-weight and three-weight codes derived from quadratic functions. Tang and Li et~al.~\cite{tang2016} generalized these constructions using weakly regular bent functions and evaluated the corresponding weight distributions through cyclotomic-field methods. Ding, Heng, and Zhou~\cite{DingHengZhou2018} developed a framework for minimal binary linear codes and determined their complete weight distributions using Krawtchouk polynomials.

The sets of codewords of a fixed weight in a linear code play a crucial role in coding theory. Minimum-weight codewords are essential for decoding and for establishing bounds on the error-correcting capability of the code. From a combinatorial perspective, the supports of constant-weight codewords can form the blocks of combinatorial \(t\)-designs, providing a fundamental connection between linear codes and design theory. The classical Assmus-Mattson theorem~\cite{Assmus1969} provides sufficient conditions under which the supports of all codewords of a given weight form a \(t\)-design, but many codes supporting \(t\)-designs fall outside its scope.
Ding and his collaborators have made significant progress in extending this connection~\cite{Tang2019}. Ding and Tang~\cite{DingTang2020} constructed infinite families of near-MDS codes whose constant-weight codes support nontrivial \(t\)-designs. Tang and Ding~\cite{TangDing2021} further constructed infinite families of linear codes supporting 4-designs.

In this paper, we conduct a detailed investigation of the code and the geometric structures associated with the sets of constant-weight codewords in linear codes. Let \( \mathcal{C} \) be an \([n,k]\) linear code, and denote by \( L(\mathcal{C}) \) the set of all linear functions defined on \( \mathcal{C} \). For each weight \( w \) such that \( A_w \neq 0 \), let \( \mathrm{wt}^{-1}_{\mathcal{C}}(w) \) denote the set of all codewords in \( \mathcal{C} \) with Hamming weight \( w \). When \(\mathcal{C}\) is clear from the context, we often abbreviate \( \mathrm{wt}^{-1}_{\mathcal{C}}(w) \) as \( \mathrm{wt}^{-1}(w) \).
We define \( \mathcal{C}^*(w) \) as the linear code corresponding to the defining set \( \mathrm{wt}^{-1}(w) \), that is,
\begin{eqnarray}\label{eq:C*}
\mathcal{C}^*(w) = \left\{ \left(\ell(\mathbf{c})\right)_{\mathbf{c} \in \mathrm{wt}^{-1}(w)} : \ell \in L(\mathcal{C}) \right\}.
\end{eqnarray}
The objective of this paper is to construct linear codes \( \mathcal{C}^*(w) \) from certain well-known linear codes \( \mathcal{C} \). If \( \mathcal{C} \) is suitably chosen, the resulting code \( \mathcal{C}^*(w) \) may possess good or even optimal parameters. We establish some general structural properties of \( \mathcal{C}^*(w) \), and in the case where \( \mathcal{C} \) is a two-weight code, we completely determine \( \mathcal{C}^*(w) \). Moreover, we illustrate applications of this method in the study of two-weight codes, including the derivation of upper bounds on the minimum weight and lower bounds on the maximum weight, as well as results concerning the non-extendibility of two-weight codes.

Section~\ref{sec:auxiliary} provides an overview of relevant facts concerning linear codes and combinatorial \(t\)-designs. Section~\ref{generic} presents some general results on the linear codes associated with the constant-weight codewords of a linear code. Section~\ref{two-weight} gives the main parameters of the linear codes derived from the constant-weight codewords of two-weight codes and illustrates the applications of this new perspective in the study of two-weight codes. Finally, Section~\ref{conc} concludes the paper.

\section{Some auxiliary results}\label{sec:auxiliary}

In this section, we introduce several foundational concepts that are essential for understanding the results presented in this paper, specifically focusing on linear codes, generalized Hamming weights, combinatorial $t$-designs, and their interrelations.

\subsection{Linear codes and Pless Power Moments}

Let \( D = \{\{ \mathbf{v_1}, \mathbf{v_2}, \dots, \mathbf{v_n} \}\} \) be a multiset of vectors in a vector space \( V \) over  \( \mathbb{F}_q \), allowing repeated elements. Let \( L(V) \) denote the set of all linear functionals on \( V \). A linear code of length \( n \) over \( \mathbb{F}_q \) is defined as
\[
\mathcal{C}_D = \left\{ \left( \ell(\mathbf{v_1}), \ell(\mathbf{v_2}), \dots, \ell(\mathbf{v_n}) \right) : \ell \in L(V) \right\},
\]
where \( D \) is referred to as the \emph{defining set} of the code \( \mathcal{C}_D \). By construction, the dimension of the code \( \mathcal{C}_D \) is at most \( \dim(V) \).
This framework is quite versatile, enabling the construction of a broad class of linear codes through various choices of defining sets \( D  \). The choice of \( D \) influences the parameters and properties of the resulting code, making this a powerful method for generating codes with specific desired characteristics. In particular, many classes of linear codes with a few weights have been constructed using this technique \cite{Ding2015,tang2016,zhou2016}. If \( D \) is a subset of the projective space \( \mathrm{PG}(k-1, q) \), let \( \widetilde{D} \) be any lifting of \( D \) to the vector space \( \mathbb{F}_q^k \). It is clear that the linear codes \( \mathcal{C}_{\widetilde{D}} \) are equivalent for different liftings \( \widetilde{D} \), and thus we refer to the code associated with \( D \) as the linear code \( \mathcal{C}_D \). The projective geometric structure of \( D \) often plays a significant role in determining the properties of the code \( \mathcal{C}_D \), particularly when \( D \) is derived from carefully chosen configurations in projective space.

In coding theory, two common operations applied to linear codes are shortening and puncturing, which are used to create new codes from existing ones, often with different lengths or dimensions.

Let \(\mathcal{C}\) be a linear code of length \(n\), dimension \(k\), and minimum distance \(d\) over the finite field \(\gf_q\). Given a set \(T\) of \(t\) coordinates, \textit{puncturing} \(\mathcal{C}\) is the process of deleting all coordinates in \(T\) from each codeword of \(\mathcal{C}\). The resulting code, which is still linear, has length \(n - t\) and is referred to as the \textit{punctured code} of \(\mathcal{C}\), denoted by \(\mathcal{C}^T\).
Furthermore, let \(\mathcal{C}(T)\) be the set of codewords in \(\mathcal{C}\) that are zero in all coordinates of \(T\). This subset \(\mathcal{C}(T)\) forms a subcode of \(\mathcal{C}\). By puncturing \(\mathcal{C}(T)\) on \(T\), we obtain another linear code of length \(n - t\), known as the \textit{shortened code} of \(\mathcal{C}\), denoted by \(\mathcal{C}_T\).

Under certain conditions, the dimension of a shortened code \(\mathcal{C}_T\) is known and given in the
next proposition \cite{Huffman}.
\begin{proposition}\label{prop:punc-shor}
Let \(\mathcal{C}\) be an \([n, k, d]\) code over \(\gf_q\) and let \(T\) be any set of \(t\) coordinates. Let \(\mathcal{C}^T\) denote the punctured code of \(\mathcal{C}\) in all coordinates in \(T\). Then the following hold:
\begin{enumerate}
    \item[(1)] \((\mathcal{C}^{\perp})_T = (\mathcal{C}^T)^{\perp}\) and \((\mathcal{C}^{\perp})^T = (\mathcal{C}_T)^{\perp}\).
    \item[(2)] If \(t < d\), then \(\mathcal{C}^T\) and \((\mathcal{C}^{\perp})_T\) have dimensions \(\kappa\) and \(n - k-t\), respectively.
    \item[(3)] If \(t = d\) and \(T\) is the set of coordinates where a minimum weight codeword is nonzero, then \(\mathcal{C}^T\) and \((\mathcal{C}^{\perp})_T\) have dimensions \(k - 1\) and \(n - k - t + 1\), respectively.
\end{enumerate}
\end{proposition}

In contrast to puncturing, if there exists an \((n+1, k, d+1)\) code \( \mathcal{C}' \) from which \( \mathcal{C} \) can be derived by deleting a coordinate, then \( \mathcal{C} \) is termed \emph{extendable} to \( \mathcal{C}' \), and \( \mathcal{C}' \) is referred to as an \emph{extension} of \( \mathcal{C} \). Moreover, \( \mathcal{C} \) is considered \emph{doubly extendable} if it has an extension that itself possesses an additional extension.

The \emph{MacWilliams identities} and the \emph{Pless power moments} (see \cite{Huffman}) are essential tools for studying the relationship between the weight enumerators of a linear code and its dual, as well as for understanding the distribution of codeword weights.

Let $\C$ be an $[n, k]_q$ linear code over a finite field $\mathbb{F}_q$, and let $\C^\perp$ denote its dual code. A projective code is a linear code such that the minimum weight of its dual code is at least
three.
The weight enumerator of $\C$ is given by
\[
W_\C(x, y) = \sum_{i=0}^{n} A_i x^{n-i} y^i,
\]
where $A_i$ denotes the number of codewords in $\C$ of Hamming weight $i$. 
The MacWilliams identity relates the weight enumerator of a linear code $\C$ to that of its dual code $\C^\perp$ as follows
\[
W_{\C^\perp}(x, y) = \frac{1}{|\C|} W_\C(x + (q-1)y, x - y),
\]
where $W_{\C}(x, y)$ represents the weight enumerator of code $\C$, and $W_{\C^\perp}(x, y)$ represents the weight enumerator of its dual code. This identity allows the computation of the weight enumerator of the dual code directly from the weight distribution of the original code.
By comparing the coefficients of the polynomials on both sides of the above equation, a more direct relationship between the weight distribution of $\C$ and  $\C^\perp$ can be obtained
\begin{eqnarray}\label{eq:A-Adual}
\frac{1}{q^k} \sum_{i=0}^{n-r} \binom{n-i}{r} A_i = \frac{1}{q^r} \sum_{i=0}^{r} \binom{n-i}{n-r} A_i^\perp,	
\end{eqnarray}
where $0\le r \le n$ and $A_i^\perp$ represents the number of codewords of weight $i$ in the dual code $\C^\perp$. These equations further reveal how the weight distribution of the original code can be used to compute the weight distribution of the dual code.

The \emph{Pless power moments}, derived from the MacWilliams identity, relate the moments of the weight distribution of codewords in $\C$ to the weight distribution of its dual code $\C^\perp$ \cite{DingTang2022}.
The first few Pless power moment is expressed as
\[
\left\{
\begin{aligned}
    \sum_{i=0}^{n} A_i &= q^k, \\
    \sum_{i=0}^{n} A_i i &= q^{k-1} \left(qn - n - A_1^\perp \right), \\
    \sum_{i=0}^{n} A_i i^2 &= q^{k-2} \left[ (q-1)n(qn - n + 1) - (2qn - q - 2n + 2)A_1^\perp + 2A_2^\perp \right], \\
    \sum_{i=0}^{n} A_i i^3 &= q^{k-3} \left[ (q-1)n(q^2 n^2 - 2qn^2 + 3qn - q + n^2 - 3n + 2) \right. \\
    & \quad \left. - (3q^2 n^2 - 3q^2 n - 6qn^2 + 12qn + q^2 - 6q + 3n^2 - 9n + 6)A_1^\perp \right. \\
     & \quad \left.  + 6(qn - q - n + 2)A_2^\perp - 6A_3^\perp \right].
\end{aligned}
\right.
\]

In the binary case ($q=2$), these moments simplify to:
\[
\left\{
\begin{aligned}
    \sum_{i=0}^{n} A_i &= 2^k, \\
    \sum_{i=0}^{n} A_i i &= 2^{k-1}(n - A_1^\perp), \\
    \sum_{i=0}^{n} A_i i^2 &= 2^{k-2}[n(n + 1) - 2nA_1^\perp + 2A_2^\perp], \\
    \sum_{i=0}^{n} A_i i^3 &= 2^{k-3}[n(n + 3) - (3n^2 + 3n - 2)A_1^\perp + 6nA_2^\perp - 6A_3^\perp].
\end{aligned}
\right.
\]
These equations represent the first four Pless power moments. The first moment counts the total number of codewords, the second moment sums the codeword weights, while the third and fourth moments provide more detailed information about the weight distribution, including variance and higher-order characteristics.

The MacWilliams identities and Pless power moments are widely used in coding theory to study the properties of dual codes, analyze error-correcting performance, and derive bounds on parameters such as the minimum distance. These tools offer deep insights into the structure of linear codes and their duals, playing a crucial role in both theoretical and practical applications.

\subsection{Generalized Hamming weights}
Generalized Hamming weights are an important concept in coding theory, extending the idea of traditional Hamming weights to higher-dimensional subcodes. They play a key role in understanding the structure and error-correcting capabilities of linear codes, particularly in applications such as wiretap channels, network coding, and secret sharing.

Let $\C$ be an code of length $n$ and dimension $k$ over $\mathbb{F}_q$. The $r$-th generalized Hamming weight $d_r(\C)$ for $1 \leq r \leq k$ is given by
\[
d_r(\C) = \min \{ \#\mathrm{Supp}(\mathcal D) : \mathcal D \text{ is an } r\text{-dimensional subcode of } \C \},
\]
where $\mathrm{Supp}(\mathcal D) $, the support of $\mathcal D$, is the set of coordinate positions for which at least one word of $\mathcal D$ has a non-zero coordinate. The quantity 
\[
\mathrm{wt}(\mathcal D) = \# \mathrm{Supp}(\mathcal D) 
\]
is called the effective length or the weight of $\mathcal D$.

The $r$-dimensional weight distribution of $\C$ consists of the numbers 
\[
A^{(r)}(i) = \# \left \{ \mathcal D : \mathrm{wt}(\mathcal D) = i, \dim(\mathcal D) = r \right \},
\]
where $\mathcal D$ runs through the $r$-dimensional subcodes of $\C$ and $0 \le i \le n$. Since the projection of $\mathcal D$ onto a coordinate position is an $\mathbb{F}_q$-linear map, we immediately see that
\[
\mathrm{wt}(\mathcal D) = \frac{1}{q^r - q^{r-1}} \sum_{\bf c \in \mathcal D} \wt(\bf c).
\]

\subsection{Combinatorial $t$-Designs and the Assmus-Mattson Theorem}

A \textit{combinatorial $t$-$(v,k,\lambda)$ design} is a pair $(\mathcal{P}, \mathcal{B})$, where $\mathcal{P}$ is a $v$-element set of \textit{points}, and $\mathcal{B}$ is a multiset of $k$-subsets of $\mathcal{P}$, called \textit{blocks}, such that every $t$-subset of $\mathcal{P}$ is contained in exactly $\lambda$ blocks. A design is called \textit{simple} if there are no repeated blocks, and is called a \textit{Steiner system}, denoted by $S(t,k,v)$, if $t \geq 2$ and $\lambda = 1$.

Given a $[v,k,d]$ linear code $\C$ over the finite field $\mathrm{GF}(q)$, with weight distribution $(A_0, A_1, \ldots, A_v)$ and weight enumerator $\sum_{i=0}^v A_i z^i$, it is often of interest to determine whether the code induces combinatorial designs. Let $\mathcal{P}(\C) = \{0,1,\dots,v-1\}$ be the set of coordinate positions of $\C$. For a codeword $\mathbf{c} = (c_0,\ldots,c_{v-1}) \in \C$, its support is defined as
\[
\mathrm{Supp}(\mathbf{c}) = \{ i \in \mathcal{P}(\C) : c_i \neq 0 \}.
\]
For a fixed weight $w$, we consider the multiset
\[
\mathcal{B}_w(\C) = \frac{1}{q-1} \left\{\left\{ \mathrm{Supp}(\mathbf{c}) : \mathbf{c} \in \C,~\mathrm{wt}(\mathbf{c}) = w \right\}\right\},
\]
where $\{\{\cdot\}\}$ denotes multiset notation and $\frac{1}{q-1}$ scales the multiplicity of each block accordingly. Under appropriate conditions, the pair $(\mathcal{P}(\C), \mathcal{B}_w(\C))$ forms a $t$-design with
\[
b = \frac{1}{q-1}A_w, \quad \lambda = \frac{\binom{w}{t}}{(q-1)\binom{v}{t}}A_w.
\]
Such designs may be trivial, simple, or possess repeated blocks, depending on the structure of $\C$.

A powerful tool for constructing such designs is the classical Assmus-Mattson Theorem \cite{Assmus1969}, stated as follows:

\begin{theorem}[Assmus-Mattson]
Let $\C$ be a linear code over $\mathrm{GF}(q)$ with length $\nu$ and minimum distance $d$, and let $\C^\perp$ denote the dual code with minimum distance $d^\perp$. Suppose $1 \le t < \min\{d, d^\perp\}$, and that there are at most $d^\perp - t$ nonzero weights of $C$ in $\{1,2,\ldots,\nu - t\}$. Then for all $k \in \{0, 1, \ldots, \nu\}$, the codewords of weight $k$ in $\C$ and $\C^\perp$ support $t$-designs. That is, both $(\mathcal{P}(\C), \mathcal{B}_k(\C))$ and $(\mathcal{P}(\C^\perp), \mathcal{B}_k(\C^\perp))$ form $t$-designs.
\end{theorem}

While the classical Assmus-Mattson theorem provides a useful criterion for determining when the supports of codewords in a linear code form $t$-designs, this condition is not sufficient in general and sometimes fails to explain why certain codes give rise to $t$-designs. To address this limitation, several refined generalizations have been developed. One such result, presented in \cite{Tang2019}, considerably broadens the applicability of the Assmus--Mattson theorem and is stated as follows:

\begin{theorem}[A Generalization of the Assmus-Mattson Theorem]\label{thm:AGM}
Let $\C$ be a linear code over $\mathrm{GF}(q)$ of length $\nu$ with minimum distance $d$, and let $C^\perp$ be its dual with minimum distance $d^\perp$. Let $s$ and $t$ be positive integers with $t < \min\{d, d^\perp\}$. Let $S$ be an $s$-subset of $\{d, d+1, \ldots, \nu - t\}$. Suppose that for all $\ell \in \{d, d+1, \ldots, \nu - t\} \setminus S$ and $0 \le \ell^\perp \le s + t - 1$, the supports of codewords of weight $\ell$ in $\C$ and of weight $\ell^\perp$ in $\C^\perp$ form $t$-designs. Then both $\C$ and $\C^\perp$ support $t$-designs for all weights $k \ge t$.
\end{theorem}

This broader criterion allows for more flexibility in verifying the $t$-design property and reveals deeper connections between the algebraic structure of linear codes and combinatorial design theory. These results not only provide constructive techniques for generating designs from codes but also contribute to the classification and understanding of code symmetries and automorphism groups.

\section{A generic construction of linear codes from weight functions} \label{generic}
In this section, a general second construction of a linear code is presented, derived from the set of codewords of a given linear code having a fixed weight. The dimension, number of weights, and weight distribution of the constructed code are further investigated. Moreover, connections are established between the existence of codewords in the constructed code whose weights equal the code length, blocking sets, and the extendability of the original code.

Let $\C$ be an $[n,k]_q$ code, and suppose its codewords are represented as row vectors of length $n$. Let $A_w$ denote the number of codewords in $\C$ with Hamming weight $w$, and let $\mathbf{c}_1, \mathbf{c}_2, \dots, \mathbf{c}_{A_w}$ represent all such codewords of weight $w$. Define $M_w$ to be an $n \times A_w$ matrix as follows:
\begin{eqnarray}\label{eq:Matrix}
M_w = \begin{bmatrix}
\mathbf{c}_1^{T} & \mathbf{c}_2^{T} & \cdots & \mathbf{c}_{A_w}^{T}
\end{bmatrix},	
\end{eqnarray}
where each $\mathbf{c}_i^T$ is the transpose of the row vector $\mathbf{c}_i$.
It is easy to observe that the linear code $\C^*(w)$, as defined in (\ref{eq:C*}), is precisely the linear code spanned by the row vectors of the matrix $M_w$. 

Let $G$ be a generator matrix of the code $\C^*(w)$. By the definition of $\C^*(w)$, if $\mathbf{v}$ is a column vector of $G$, then $a\mathbf{v}$ is also a column vector of $G$ for any $a \in \mathbb{F}_q^*$. Therefore, without loss of generality, $G$ can be written in the following block form:
\[
G = \left[ G_0 \mid \alpha G_0 \mid \alpha^2 G_0 \mid \cdots \mid \alpha^{q-2} G_0 \right],
\]
where $G_0$ is a submatrix of $G$, and $\alpha$ is a generator of the multiplicative group $\mathbb{F}_q^*$. 
Let $\overline{\C^*(w)}$ be the linear code generated by $G_0$. It is evident that, up to equivalence, $\overline{\C^*(w)}$ does not depend on the specific choice of $G_0$. We refer to this code as the projectivization of $\C^*(w)$. Clearly, if $\C^*(w)$ is an $[n, k, d]$ code, then $\overline{\C^*(w)}$ is a projective $\left[\frac{n}{q-1}, k, \frac{d}{q-1}\right]$ code.

The following example illustrates that the proposed construction method naturally yields optimal linear codes.
\begin{example}
Let $\mathcal{C}$ be the ternary Golay code with parameters $[11,6,5]$, whose generator matrix is given by
\[
G = 
\left[
\begin{array}{ccccccccccc}
1 & 0 & 0 & 0 & 0 & 0 & 2 & 0 & 1 & 2 & 1 \\
0 & 1 & 0 & 0 & 0 & 0 & 1 & 2 & 2 & 2 & 1 \\
0 & 0 & 1 & 0 & 0 & 0 & 1 & 1 & 1 & 0 & 1 \\
0 & 0 & 0 & 1 & 0 & 0 & 1 & 1 & 0 & 2 & 2 \\
0 & 0 & 0 & 0 & 1 & 0 & 2 & 1 & 2 & 2 & 0 \\
0 & 0 & 0 & 0 & 0 & 1 & 0 & 2 & 1 & 2 & 2
\end{array}
\right].
\]
The weight enumerator of $\mathcal{C}$ is
\[
W_{\mathcal{C}}(x,y) = 1 + 132x^6y^5 + 132x^5y^6 + 330x^3y^8 + 110x^2y^9 + 24y^{11}.
\]
Then, the code $\overline{\mathcal{C}^*(9)}$ is a $[55,5,36]$ code with weight enumerator
\[
W_{\overline{\mathcal{C}^*(9)}}(x,y) = 1 + 220x^{19}y^{36} + 22x^{10}y^{45},
\]
which is optimal according to the Griesmer bound.
Similarly, the code $\overline{\mathcal{C}^*(11)}$ is a $[12,6,6]$ code with weight enumerator
\[
W_{\overline{\mathcal{C}^*(11)}}(x,y) = 1 + 264x^6y^6 + 440x^3y^9 + 24y^{12},
\]
which also meets the Griesmer bound and is, in fact, equivalent to the extended ternary Golay code.
\end{example}

The following theorem provides a more explicit characterization of the linear code $\overline{\mathcal{C}^*(w)}$.

\begin{theorem}\label{thm:def-C-w}
Let $\mathcal{C}$ be a projective $[n, k]$ linear code over $\mathbb{F}_q$, and suppose that the codewords of $\mathcal{C}$ with Hamming weight $w$ span the entire code $\mathcal{C}$. Let $G = [\mathbf{g}_1~\mathbf{g}_2~\cdots~\mathbf{g}_n]$ be a generator matrix of $\mathcal{C}$, where each $\mathbf{g}_i \in \mathbb{F}_q^k$ is a column vector. Define the set
\[
D = \left\{ \mathbf{v} \in \mathbb{F}_q^k : \mathrm{wt}\left( \left( \langle \mathbf{v}, \mathbf{g}_i \rangle \right)_{1 \le i \le n} \right) = w \right\},
\]
and let $\overline{D}$ denote the corresponding projective point set, obtained by selecting a single representative from each 1-dimensional subspace spanned by vectors in $D$.
Then,  $\overline{\mathcal{C}^*(w)}$ is monomially equivalent to the linear code $\mathcal{C}_{\overline{D}}$ defined by the defining set $\overline{D}$.
\end{theorem}

\begin{proof}
First, note that the code $\mathcal{C}$ can be described as
\[
\mathcal{C} = \left\{ \left( \langle \mathbf{v}, \mathbf{g}_i \rangle \right)_{1 \le i \le n} : \mathbf{v} \in \mathbb{F}_q^k \right\}.
\]
By the definition of $D$, the linear map from $\mathbb{F}_q^k$ to $\mathcal{C}$ given by
\[
\mathbf{v} \mapsto \left( \langle \mathbf{v}, \mathbf{g}_i \rangle \right)_{1 \le i \le n}
\]
restricts to a linear isomorphism between $D$ and $\mathrm{wt}^{-1}(w)$, the subset of $\mathcal{C}$ consisting of all codewords of weight $w$. Therefore, the linear codes $\mathcal{C}_D$ and $\mathcal{C}^*(w)$ are monomially equivalent.
Passing to the projective versions, the projective code $\mathcal{C}_{\overline{D}}$ defined by $\overline{D}$ is also monomially equivalent to $\overline{\mathcal{C}^*(w)}$.
\end{proof}

The following theorem demonstrates a form of duality in the proposed construction. Under certain conditions, applying the construction twice returns a code that is equivalent to the original one.

\begin{theorem}\label{thm:C**-C}
Let $\mathcal{C}$ be a projective $[n, k]$ linear code over $\mathbb{F}_q$ and let
$
\lambda = \frac{w A_w(\mathcal{C})}{n(q - 1)}.
$
Assume that the codewords of $\mathcal{C}$ with Hamming weight $w$ span the entire code $\mathcal{C}$, and that the supports of the weight-$w$ codewords form a $1$-design (possibly with repeated blocks). Suppose further that
$
A_{\lambda}(\overline{\mathcal{C}^*(w)}) = n(q - 1).
$
Then, $\overline{ \overline{\mathcal{C}^*(w)}^*(\lambda) }$ is monomially equivalent to the original code $\mathcal{C}$.
\end{theorem}

\begin{proof}

Let $\mathbf{g}_i$ ($1 \leq i \leq n$), $D$, and $\overline{D}$ be defined as in Theorem~\ref{thm:def-C-w}.  
Denote by $m(q-1)$ the number of codewords of $\mathcal{C}$ with Hamming weight $w$.  
Then there exist $m$ pairwise non-collinear vectors $\mathbf{v}_1, \dots, \mathbf{v}_m \in \mathbb{F}_q^k$ such that $\overline{D}$ is precisely the projective point set generated by them.  
Consequently, in Fig.~\ref{Fig:Codewords-C-w}, the $m$ rows represent the direction set of all weight-$w$ codewords of $\mathcal{C}$.  
Since the weight-$w$ codewords span the whole code $\mathcal{C}$, the vectors $\mathbf{v}_1, \dots, \mathbf{v}_m$ span $\mathbb{F}_q^k$.  
By Theorem~\ref{thm:def-C-w}, the set $\{\mathbf{v}_1, \dots, \mathbf{v}_m\}$ can thus serve as a defining set for $\overline{\mathcal{C}^*(w)}$.  
Furthermore, because the supports of the weight-$w$ codewords form a $1$-design, each column in Fig.~\ref{Fig:Codewords-C-w} corresponds to a codeword of $\overline{\mathcal{C}^*(w)}$ with weight $\lambda$.  
Therefore, the columns of Fig.~\ref{Fig:Codewords-C-w} represent exactly the direction set of the weight-$\lambda$ codewords of $\overline{\mathcal{C}^*(w)}$.  
Applying Theorem~\ref{thm:def-C-w} once again, we deduce that $\overline{ \overline{\mathcal{C}^*(w)}^*(\lambda) }$ is defined by the set $\{ \mathbf{g}_1, \dots, \mathbf{g}_n\}$.  
Hence, $\overline{ \overline{\mathcal{C}^*(w)}^*(\lambda) }$ is monomially equivalent to the original code $\mathcal{C}$.

	\begin{figure}[htbp]
  \centering
  \begin{center}
\begin{adjustbox}{max width=\textwidth}
\begin{tikzpicture}[scale=1, every node/.style={font=\footnotesize}]
  \foreach \i in {0,...,6} {
    \draw (0,-\i) -- (6,-\i);
  }

  \foreach \j in {0,...,6} {
    \draw (\j,0) -- (\j,-6);
  }

  \node at (0.5, 0.5) {$\mathbf{g}_1$};
  \node at (1.5, 0.5) {$\mathbf{g}_2$};
  \node at (2.5, 0.5) {$\dots$};
  \node at (3.5, 0.5) {\textcolor{red}{$\mathbf{g}_j$}};
  \node at (4.5, 0.5) {$\cdots$};
  \node at (5.5, 0.5) {$\mathbf{g}_n$};

  \node at (-0.5, -0.5) {$\mathbf{v}_1$};
  \node at (-0.5, -1.5) {$\mathbf{v}_2$};
  \node at (-0.5, -2.5) {$\vdots$};
  \node at (-0.5, -3.5) {\textcolor{green}{$\mathbf{v}_i$}};
  \node at (-0.5, -4.5) {$\vdots$};
  \node at (-0.5, -5.5) {$\mathbf{v}_m$};

  \node at (0.5, -0.5) {$\langle \mathbf{v}_1, \mathbf{g}_1 \rangle$};
  \node at (1.5, -0.5) {$\langle \mathbf{v}_1, \mathbf{g}_2 \rangle$};
  \node at (2.5, -0.5) {$\dots$};
  \node at (3.5, -0.5) {\textcolor{red}{$\langle \mathbf{v}_1, \mathbf{g}_j \rangle$}};
  \node at (4.5, -0.5) {$\cdots$};
  \node at (5.5, -0.5) {$\langle \mathbf{v}_1, \mathbf{g}_n \rangle$};

  \node at (0.5, -1.5) {$\langle \mathbf{v}_2, \mathbf{g}_1 \rangle$};
  \node at (1.5, -1.5) {$\langle \mathbf{v}_2, \mathbf{g}_2 \rangle$};
  \node at (2.5, -1.5) {$\dots$};
  \node at (3.5, -1.5) {\textcolor{red}{$\langle \mathbf{v}_2, \mathbf{g}_j \rangle$}};
  \node at (4.5, -1.5) {$\cdots$};
  \node at (5.5, -1.5) {$\langle \mathbf{v}_2, \mathbf{g}_n \rangle$};

  \node at (0.5, -2.5) {$\vdots$};
  \node at (1.5, -2.5) {$\vdots$};
  \node at (2.5, -2.5) {$\ddots$};
  \node at (3.5, -2.5) {\textcolor{red}{$\vdots$}};
  \node at (4.5, -2.5) {$\cdots$};
  \node at (5.5, -2.5) {$\vdots$};

  \node at (0.5, -3.5) {\textcolor{green}{$\langle \mathbf{v}_i, \mathbf{g}_1 \rangle$}};
  \node at (1.5, -3.5) {\textcolor{green}{$\langle \mathbf{v}_i, \mathbf{g}_2 \rangle$}};
  \node at (2.5, -3.5) {\textcolor{green}{$\dots$}};
  \node at (3.5, -3.5) {\textcolor{green}{$\langle \mathbf{v}_i,$}\textcolor{red}{$\mathbf{g}_j \rangle$}};
  \node at (4.5, -3.5) {\textcolor{green}{$\cdots$}};
  \node at (5.5, -3.5) {\textcolor{green}{$\langle \mathbf{v}_i, \mathbf{g}_n \rangle$}};
  \node at (7.1, -3.5) {\textcolor{green}{Row weight:~$w$}};

  \node at (0.5, -4.5) {$\vdots$};
  \node at (1.5, -4.5) {$\vdots$};
  \node at (2.5, -4.5) {$\cdots$};
  \node at (3.5, -4.5) {\textcolor{red}{$\vdots$}};
  \node at (4.5, -4.5) {$\ddots$};
  \node at (5.5, -4.5) {$\vdots$};

  \node at (0.5, -5.5) {$\langle \mathbf{v}_m, \mathbf{g}_1 \rangle$};
  \node at (1.5, -5.5) {$\langle \mathbf{v}_m, \mathbf{g}_2 \rangle$};
  \node at (2.5, -5.5) {$\dots$};
  \node at (3.5, -5.5) {\textcolor{red}{$\langle \mathbf{v}_m, \mathbf{g}_j \rangle$}};
  \node at (4.5, -5.5) {$\cdots$};
  \node at (5.5, -5.5) {$\langle \mathbf{v}_m, \mathbf{g}_n \rangle$};
  \node at (3.5, -6.3) {\textcolor{red}{Column weight:~$\lambda$}};
\end{tikzpicture}
\end{adjustbox}
  \end{center}
  \caption{Codewords of weight $\lambda$ in $\overline{\mathcal{C}^*(w)}$
}\label{Fig:Codewords-C-w}
\end{figure}

	\end{proof}

The following proposition establishes the relationship between the weight of subcodes in $\C^*(w)$ and the number of codewords of weight $w$ in subcodes of $\C$.
\begin{proposition}\label{prop:wt=Aw}
	Let $\C$ be an $[n,k]$ code with $A_{w} \neq 0$. Let $\mathcal D$ be an $r$-dimensional subcode of $\C^*(w)$. Let $E: L(\C) \to \C^*(w)$ be the linear map given by 
$\ell \mapsto \left( \ell(\mathbf{c}) \right)_{\mathbf{c} \in \wt^{-1}(w)}$. Let $\ker \left( E^{-1}(\mathcal{D}) \right)$ be the set consisting of all $\mathbf{c} \in \C$ such that $\ell(\mathbf{c}) = 0$ for all $\ell \in E^{-1}(\mathcal{D})$. Then $\ker \left( E^{-1}(\mathcal{D}) \right)$ is a subcode of $\C$ with codimension $\dim (\C)-\dim \left ( \ker ( E) \right ) - r$. Furthermore, 
\[
\wt(\mathcal{D}) = A_w(\C) - A_w\left( \ker \left( E^{-1}(\mathcal{D}) \right) \right),
\]
where \(A_w(\C)\) represents the number of codewords of weight \(w\) in \(\C\), and \(A_w\left( \ker \left( E^{-1}(\mathcal{D}) \right) \right)\) denotes the number of codewords of weight \(w\) in the subcode \(\ker \left( E^{-1}(\mathcal{D}) \right)\) of $\C$.
\end{proposition}

\begin{proof}
	According to the definition of $E$, the diagram below is commutative, with both rows forming exact sequences.
\[
\begin{tikzcd}
0\arrow[r]&\ker (E) \arrow[r] \arrow[d, equal] & E^{-1}(\mathcal{D}) \arrow[r, "E|_{E^{-1}(\mathcal{D})}"] \arrow[d, hook] &  \mathcal{D} \arrow[d, hook] \arrow[r]  & 0 \\
0\arrow[r]& \ker (E) \arrow[r] & L(\C) \arrow[r, "E"'] & \C^*(w) \arrow[r]  & 0 
\end{tikzcd}.
\]
From this we deduce that 
\[\dim (E^{-1} (\mathcal D))=\dim (\ker (E)) + \dim (\mathcal D).\]
Consequently, \( \ker\left( E^{-1}(\mathcal{D}) \right) \) is a linear subspace of \( \mathcal{C} \), and its dimension is given by
\[
\begin{aligned}
\dim\left( \ker\left( E^{-1}(\mathcal{D}) \right) \right) 
&= \dim(\mathcal{C}) - \dim\left( E^{-1}(\mathcal{D}) \right) \\
&= \dim(\mathcal{C}) - \dim\left( \ker(E) \right) - r.
\end{aligned}
\]
By the definition of \( \mathrm{Supp}(\mathcal{D}) \), we have
\[
\mathrm{Supp}(\mathcal{D}) = \mathrm{wt}^{-1}_{\mathcal{C}}(w) \setminus \mathrm{wt}^{-1}_{\ker\left( E^{-1}(\mathcal{D}) \right)}(w).
\]
Therefore,
\[
\mathrm{wt}(\mathcal{D}) = A_w(\mathcal{C}) - A_w\left( \ker \left( E^{-1}(\mathcal{D}) \right) \right),
\]
which is precisely the desired conclusion.

\end{proof}

Let \( \mathcal{C} \) be an \([n, k]\) linear code over the finite field \( \mathbb{F}_q \).
Let \( \left[ \begin{smallmatrix} \mathcal{C} \\ r \end{smallmatrix} \right]_q \) denote the set of all \( r \)-dimensional subcodes of a linear code \( \mathcal{C} \subseteq \mathbb{F}_q^n \). The cardinality of \( \left[ \begin{smallmatrix} \mathcal{C} \\ r \end{smallmatrix} \right]_q \) is given by the Gaussian binomial coefficient \( \left[ \begin{smallmatrix} \dim(\mathcal{C}) \\ r \end{smallmatrix} \right]_q \), which counts the number of \( r \)-dimensional subspaces in a \( \dim(\mathcal{C}) \)-dimensional vector space over \( \mathbb{F}_q \).
The following corollary shows that studying the value distribution of the weight function on \( \left[ \begin{smallmatrix} \C^*(w) \\ r \end{smallmatrix} \right]_q \) is equivalent to studying the value distribution of the function on  \( \left[ \begin{smallmatrix} \mathcal{C} \\  k-r \end{smallmatrix} \right]_q \) given by $\mathcal{B} \mapsto A_w(\mathcal{B})$ for the case $\dim(\C^*(w)) = \dim(\C)$, where \( \left[ \begin{smallmatrix} \mathcal{C} \\  k-r \end{smallmatrix} \right]_q \).

\begin{corollary}\label{cor:C*-distribution}
	Let $\C$ be an $[n,k]$ code with $A_{w} \neq 0$. Suppose that $\dim(\C^*(w))= \dim(\C) $. Then the $r$-th generalized Hamming weight $d_r \left( \C^*(w) \right)$ is equal to
\[
A_w(\C) - \max \left\{ A_w(\mathcal{B}) : \mathcal{B} \text{ is a subcode of } \C \text{ of codimension } r \right\}.\]
Moreover, the $r$-dimensional weight distribution of $\C^*(w)$ is determined by
\[
A^{(r)}_i(\C^*(w)) = \# \left\{ \mathcal{B} : A_w(\mathcal{B}) = A_w(\C) - i \text{ and } \dim(\mathcal{B}) = k - r \right\},
\]
where $\mathcal{B}$ runs over all subcodes of $\C$ of codimension $r$.

\end{corollary}

\begin{proof}
	Since \( \dim(\mathcal{C}^*(w)) = \dim(\mathcal{C}) \), it follows that \( \dim(\ker(E)) = 0 \). Consequently, the desired conclusion follows directly from Proposition \ref{prop:wt=Aw}.
\end{proof}

The following result regarding the weight distribution of the projection code \( \overline{\mathcal{C}^*(w)} \) follows directly from Corollary \ref{cor:C*-distribution}.

\begin{corollary}\label{cor:projC*-distribution}
Let \( \mathcal{C} \) be an \([n, k]\) code with \( A_w \neq 0 \), and suppose that \( \dim(\mathcal{C}^*(w)) = \dim(\mathcal{C}) \). Then, the \( r \)-th generalized Hamming weight \( d_r\left( \overline{\mathcal{C}^*(w)} \right) \) is given by
\[
d_r\left( \overline{\mathcal{C}^*(w)} \right) = \frac{1}{q-1} \left( A_w(\mathcal{C}) - \max \left\{ A_w(\mathcal{B}) : \mathcal{B} \subseteq \mathcal{C}, \ \operatorname{codim}(\mathcal{B}) = r \right\} \right).
\]
Moreover, the \( r \)-dimensional weight distribution of \( \overline{\mathcal{C}^*(w)} \) is determined by
\[
A^{(r)}_i\left( \overline{\mathcal{C}^*(w)} \right) = \# \left\{ \mathcal{B} \subseteq \mathcal{C} : \dim(\mathcal{B}) = k - r, \ A_w(\mathcal{B}) = A_w(\mathcal{C}) - (q - 1)i \right\},
\]
where the set \( \mathcal{B} \) ranges over all subcodes of \( \mathcal{C} \) with codimension \( r \).
\end{corollary}
\begin{proof}  
The conclusion of the corollary follows directly from Corollary \ref{cor:C*-distribution} and the definition of \( \overline{\mathcal{C}^*(w)} \).  
\end{proof}

The dimension of \( \mathcal{C}^*(w) \) is determined by the following proposition.
\begin{proposition}\label{prop:dim(C*)}
Let $\mathcal{C}$ be an $[n,k]$ code with $A_{w}(\C) \neq 0$, where $A_{w}(\C)$ denotes the number of codewords of weight $w$. 
Then, the dimension of $\mathcal{C}^*(w)$ is equal to the dimension of the subcode spanned by the codewords of weight $w$ in $\mathcal{C}$.  
In particular, $\mathcal{C}^*(w)$ has dimension $k$ if and only if there exists no codimension-one subcode $\mathcal{B}$ of $\mathcal{C}$ such that $A_{w}(\mathcal{B}) = A_{w}(\mathcal{C})$.
\end{proposition}

\begin{proof}
	Let \( M \) be the matrix whose rows consist of all codewords of weight \( w \) in \( \mathcal{C} \), arranged in any fixed order. Observe that \( \mathcal{C}^*(w) \) is precisely the linear code over \( \mathbb{F}_q \) generated by the column span of \( M \). Hence, the desired conclusion follows.
\end{proof}

The following proposition is fundamental for determining the weight distribution of the code \( \mathcal{C}^*(w) \).

\begin{proposition}
Let \( \mathcal{C} \) be a \( q \)-ary \([n, k]\) linear code. As the vector \( \mathbf{v} \in \mathbb{F}_q^n \setminus \C^\perp \) varies, the set
\[
\left( \mathcal{C}^\perp \cup_{\alpha \in \mathbb{F}_q^*} (\mathcal{C}^\perp + \alpha \mathbf{v}) \right)^\perp
\]
exhausts all codimension-one subcodes of \( \mathcal{C} \). In particular, when \( \mathcal{C} \) is a binary code, its codimension-one subcodes correspond one-to-one with the nonzero cosets of \( \mathcal{C}^\perp \) with nonzero weight.
\end{proposition}

\begin{proof}
Note that the dual map \( \mathcal{B} \mapsto \mathcal{B}^\perp \), where \( \mathcal{B} \subset \mathcal{C} \) and \( \dim(\mathcal{B}) = \dim(\mathcal{C}) - 1 \), establishes a one-to-one correspondence between the codimension-one subcodes of \( \mathcal{C} \) and the supercodes of \( \mathcal{C}^\perp \) of dimension \( \dim(\mathcal{C}^\perp) + 1 \). Moreover, each such supercode of \( \mathcal{C}^\perp \) can be expressed as the linear span of \( \mathcal{C}^\perp \) and a single vector \( \mathbf{v} \in \mathbb{F}_q^n \setminus \mathcal{C}^\perp \). This yields the desired result.
\end{proof}

The following result shows that the number of nonzero weights in the code \( \overline{\mathcal{C}^*(w)} \) can be bounded by the number of distinct weight distributions among the cosets of the dual code \( \mathcal{C}^\perp \).
\begin{proposition}
Let \( \mathcal{C} \) be an \([n, k, d]\) linear code, and let \( S(\mathcal{C}) \) denote the set of all nonzero weights in \( \mathcal{C} \). Assume that there are \( u^{\perp} \) distinct weight distributions among the cosets of the dual code \( \mathcal{C}^{\perp} \), and that \( \overline{\mathcal{C}^*(w)} \) has dimension \( k \), where \( w \in S(\mathcal{C}) \). Then, the code \( \overline{\mathcal{C}^*(w)} \) has at most \( u^{\perp} - 1 \) nonzero weights.
\end{proposition}
\begin{proof}
Since the cosets of the dual code \( \mathcal{C}^{\perp} \) (including \( \mathcal{C}^{\perp} \) itself) exhibit \( u^{\perp} \) distinct weight distributions, it follows from the MacWilliams identities that there exist at most \( u^{\perp} - 1 \) distinct weight distributions among the codimension-one subcodes of \( \mathcal{C} \). By applying Corollary~\ref{cor:projC*-distribution}, the code \( \overline{\mathcal{C}^*(w)} \) thus has at most \( u^{\perp} - 1 \) nonzero weights.
\end{proof}

In a projective or affine space, a blocking set with respect to \((k - s)\)-dimensional (or \(s\)-codimensional) subspaces of an \(k\)-dimensional projective or affine space is called a \(s\)-blocking set. A $1$-blocking set is often simply referred to as a blocking set. The following proposition establishes a connection between blocking sets and the partial weight distribution of \( \overline{\mathcal{C}^*(w)} \).

\begin{proposition}
Let \( \mathcal{C} \) be an \([n, k, d]\) linear code, and let \( S(\mathcal{C}) \) denote the set of all nonzero weights in \( \mathcal{C} \). Suppose that \( \overline{\mathcal{C}^*(w)} \) has dimension \( k \) for some \( w \in S(\mathcal{C}) \). Then, \(\mathrm{wt}^{-1}(w)\) is a blocking set of $\C$ if and only if \( A_{n^*} \left( \overline{\mathcal{C}^*(w)} \right) = 0 \), where \( n^* \) denotes the length of the code \( \overline{\mathcal{C}^*(w)} \).
\end{proposition}

\begin{proof}
First, suppose that \(\mathrm{wt}^{-1}(w)\) is a blocking set of \( \mathcal{C} \). Then, every codimension-one subcode \( \mathcal{B} \subset \mathcal{C} \) must contain at least one codeword of weight \( w \), i.e., \( A_w(\mathcal{B}) > 0 \). By Corollary \ref{cor:projC*-distribution}, it follows that \( A_{n^*} \bigl( \overline{\mathcal{C}^*(w)} \bigr) = 0 \).

Conversely, assume that \( A_{n^*} \bigl( \overline{\mathcal{C}^*(w)} \bigr) = 0 \). Again, by Corollary \ref{cor:projC*-distribution}, there cannot exist a codimension-one subcode \( \mathcal{B} \) of \( \mathcal{C} \) such that \( A_w(\mathcal{B}) = 0 \). Therefore, \(\mathrm{wt}^{-1}(w)\) must be a blocking set of \( \mathcal{C} \). This completes the proof.
\end{proof}

Let \( \mathcal{C} \) be an \([n, k, d]\) linear code. The extendability of \( \mathcal{C} \) is characterized by the following proposition in terms of the partial weight distribution of \( \overline{\mathcal{C}^*(d)} \).
\begin{theorem}\label{thm:ext-dist}
Let \( \mathcal{C} \) be an \([n, k, d]\) linear code. Suppose that the code \( \overline{\mathcal{C}^*(d)} \) has dimension \( k \). Then, \( \mathcal{C} \) is extendable if and only if \( A_{n^*} \left( \overline{\mathcal{C}^*(d)} \right) > 0 \), where \( n^* \) denotes the length of the code \( \overline{\mathcal{C}^*(d)} \).	
\end{theorem}

\begin{proof}
First, assume that \( \mathcal{C} \) is extendable. Then there exists an \([n+1,k,d+1]\) linear code \( \mathcal{C}' \) such that \( \mathcal{C} \) is obtained as a punctured code of \( \mathcal{C}' \). Let \( i_0 \) be a coordinate position such that \( \mathcal{C} = \mathcal{C}'^{\{i_0\}} \). It can be observed that the shortened code \( \mathcal{C}'_{i_0} \) is a subcode of \( \mathcal{C} \) with dimension \( k-1 \). Since the minimum distance of \( \mathcal{C}' \) is \( d+1 \), no codeword of weight \( d \) appears in \( \mathcal{C}'_{i_0} \). Hence, by Corollary \ref{cor:projC*-distribution}, it is obtained that \( A_{n^*} \left( \overline{\mathcal{C}^*(d)} \right) > 0 \).

Conversely, suppose that \( A_{n^*} \left( \overline{\mathcal{C}^*(d)} \right) > 0 \). Then there exists a codimension-one subcode \( \mathcal{B} \) of \( \mathcal{C} \) which contains no codewords of weight \( d \). Thus, a linear function \( f \) on \( \mathcal{C} \) can be defined such that \( \ker (f) = \mathcal{B} \). Consider the following linear code:
\[
\mathcal{C}' = \left \{ (\mathbf{c}, f(\mathbf{c})) : \mathbf{c} \in \mathcal{C} \right \}.
\]
It can be verified that \( \mathcal{C}' \) is an \([n+1,k,d+1]\) linear extension of \( \mathcal{C} \). Therefore, \( \mathcal{C} \) is extendable. This completes the proof.
\end{proof}

\section{Linear codes from the Weight Function of Projective Two-Weight Codes}\label{two-weight}

Two-weight codes are an important class of linear codes that possess exactly two distinct non-zero Hamming weights among their codewords. Such codes are of interest in coding theory due to their highly structured nature and connections to combinatorial designs, finite geometry, and cryptographic applications. In this section, we conduct a detailed study of the linear codes associated with the fixed-weight codewords of two-weight codes. We not only determine the linear structures of these derived codes completely, but also establish an upper bound on the minimum weight of two-weight codes and give a full characterization of the binary codes attaining this bound. In addition, we characterize the extendable two-weight codes and derive several divisibility properties concerning the parameters of two-weight codes.

The weight distribution of two-weight code $\C$ with $d(\C^{\perp})\ge 2$ can be derived directly from the first two Pless power moments \cite{Calderbank}.
\begin{proposition}\label{prop:wt-dist-2}
Let $\mathcal{C}$ be a two-weight $[n, k]$ code over $\gf_{q}$ with $d(\C^{\perp})\ge 2$.  Suppose that the two nonzero weights of $\mathcal{C}$ are $w_1$ and $w_2$.  Then the weight distribution of $\mathcal{C}$ is given by 
\[
A_{w_1} = \frac{1}{(w_2 - w_1)} \left( w_2 \left( q^k - 1 \right) - n q^{k-1} (q - 1) \right)
\]
and 
\[
A_{w_2} = \frac{1}{(w_1 - w_2)} \left( w_1 \left( q^k - 1 \right) - n q^{k-1} (q - 1) \right).
\]

	\end{proposition}

\subsection{Codimension-One Subcodes and the Upper Bound on the Minimum Weight of Two-Weight Codes}

The codimension-$1$ subcodes of a projective two-weight code are fully characterized as follows.
\begin{proposition}\label{prop:codimension-1}
Let $\mathcal{C}$ be a projective two-weight $[n, k]$ code over $\mathbb{F}_q$ with weights $w_1$ and $w_2$.
Let $\mathcal{B}$ be a codimension-$1$ subcode of $\mathcal{C}$. Then, the following hold:
\begin{enumerate}
    \item[(1)] The weight $\wt (\mathcal B)$ of $\mathcal{B}$ is either $n$ or $n-1$.
    
    \item[(2)] There are exactly $n$ codimension-$1$ subcodes $\mathcal{B}$ of $\mathcal{C}$ such that $\mathrm{wt}(\mathcal{B}) = n-1$.
    
    \item[(3)] Moreover, the weight distribution of $\mathcal{B}$ is as follows:
    \begin{itemize}
        \item If $\mathrm{wt}(\mathcal{B}) = n-1$, the weight distribution is given by
        \[
        A_{w_1}\left ( \mathcal B \right ) = \frac{1}{w_2 - w_1} \left( w_2 \left( q^{k-1} - 1 \right) - (n-1) q^{k-2} (q - 1) \right),
        \]
        \[
        A_{w_2} \left ( \mathcal B \right ) = \frac{1}{w_1 - w_2} \left( w_1 \left( q^{k-1} - 1 \right) - (n-1) q^{k-2} (q - 1) \right).
        \]
        
        \item If $\mathrm{wt}(\mathcal{B}) = n$, the weight distribution is given by
        \[
        A_{w_1} \left ( \mathcal B \right ) = \frac{1}{w_2 - w_1} \left( w_2 \left( q^{k-1} - 1 \right) - n q^{k-2} (q - 1) \right),
        \]
        \[
        A_{w_2} \left ( \mathcal B \right ) = \frac{1}{w_1 - w_2} \left( w_1 \left( q^{k-1} - 1 \right) - n q^{k-2} (q - 1) \right).
        \]
    \end{itemize}
\end{enumerate}
\end{proposition}

\begin{proof}
\begin{enumerate}
\item[(1)] Suppose, for the sake of contradiction, that the conclusion does not hold, so that \( \mathrm{wt}(\mathcal{B}) \le n - 2 \).  
Then there exist two coordinate positions \( i_1 \) and \( i_2 \) of \(\mathcal{C}\) such that for every codeword \( \mathbf{c} = (c_1, \ldots, c_n) \in \mathcal{B} \), we have \( c_{i_1} = c_{i_2} = 0 \).  
Consequently, the shortened code \(\mathcal{B}_{\{i_1,i_2\}}\) has the same dimension as \(\mathcal{B}\) and is contained in the shortened code \(\mathcal{C}_{\{i_1,i_2\}}\).   
It follows that \(\dim (\mathcal{C}_{\{i_1, i_2\}}) \ge \dim(\mathcal{B}_{\{i_1, i_2\}}) = k - 1\).  
On the other hand, since \(\mathcal{C}\) is a projective code, Proposition~\ref{prop:punc-shor} implies that \(\mathcal{C}_{\{i_1, i_2\}}\) must have dimension \( k - 2 \).  
This leads to a contradiction.  
Therefore, we conclude that \(\mathrm{wt}(\mathcal{B}) = n\) or \( n - 1 \).  
\item[(2)] Let \(\mathcal{C}(i_1, \ldots, i_s)\) denote the subcode of \(\mathcal{C}\) consisting of all codewords that are zero in the coordinate positions \(i_1, \ldots, i_s\).  
It is easy to see that a codimension-one subcode \(\mathcal{B}\) of \(\mathcal{C}\) satisfies \(\mathrm{wt}(\mathcal{B}) = n - 1\) if and only if \(\mathcal{B}\) is of the form \(\mathcal{C}(i)\) for some \(1 \le i \le n\).  
Moreover, it can be shown that the subcodes \(\mathcal{C}(1), \ldots, \mathcal{C}(n)\) are pairwise distinct.  
Indeed, suppose on the contrary that \(\mathcal{C}(i_1) = \mathcal{C}(i_2)\).  
Then we have \(\mathcal{C}(i_1) = \mathcal{C}(i_2) = \mathcal{C}(i_1, i_2)\).  
Consequently, \(\dim(\mathcal{C}(i_1)) = \dim(\mathcal{C}(i_1, i_2))\).  
However, by Proposition~\ref{prop:punc-shor}, we know that \(\dim(\mathcal{C}(i_1)) = k - 1\) and \(\dim(\mathcal{C}(i_1, i_2)) = k - 2\), which leads to a contradiction.  
Therefore, there are exactly \(n\) codimension-one subcodes of \(\mathcal{C}\) whose weight equals \(n - 1\).
\item[(3)] If \(\mathrm{wt}(\mathcal{B}) = n - 1\), then it is easy to see that \(\mathcal{B}\) has the same weight distribution as the shortened code \(\mathcal{C}_{\{i\}}\), where \(1 \le i \le n\).  
Moreover, it can be directly verified that \(\mathcal{C}_{\{i\}}\) is a two-weight code and that its dual satisfies \(d(\mathcal{C}_{\{i\}}^{\perp}) \ge 2\).  
The desired result then follows immediately from Proposition~\ref{prop:wt-dist-2}.  

If \(\mathrm{wt}(\mathcal{B}) = n\), then \(\mathcal{B}\) itself is a two-weight code, and its dual also satisfies \(d(\mathcal{B}^{\perp}) \ge 2\).  
Proposition~\ref{prop:wt-dist-2} then provides the weight distribution of \(\mathcal{B}\).

\end{enumerate}
\end{proof}

To describe certain two-weight codes, it is useful to invoke structural results on linear codes with constant Hamming weight. In particular, the following lemma characterizes all linear codes over a finite field in which every nonzero codeword has the same Hamming weight \cite{Bon84}.

\begin{lemma}\label{lemma:equidistant-support}
Let \( \mathbb{F}_q \) be a finite field, and let \( G \) be a \( k \times \frac{q^k - 1}{q - 1} \) matrix whose columns consist of exactly one nonzero vector from each one-dimensional subspace of \( \mathbb{F}_q^k \). Then the linear code generated by \( G \) is of dimension \( k \) and constant Hamming weight. Moreover, any other \( k \)-dimensional linear code over \( \mathbb{F}_q \) in which all nonzero codewords have the same weight is equivalent to a replication of this code, up to coordinate permutation and insertion of zero columns.
\end{lemma}

Some bounds on the nonzero weights of projective two-weight codes are presented below. Moreover, the codes that attain the upper bound on the minimum weight have been completely characterized. These results are of intrinsic significance in coding theory; more importantly, they demonstrate that analyzing the linear codes associated with codewords of a fixed weight provides a powerful tool in the study of linear codes.

\begin{theorem}\label{thm:w1-w2-bound}
Let $\mathcal{C}$ be a projective two-weight $[n, k]$ code over $\mathbb{F}_q$ with weights $w_1 < w_2$. Then
\[
w_1 \leq \frac{n(q-1)}{q} \quad \text{and} \quad w_2 \geq \frac{n(q-1)+2}{q}.
\]
Moreover, $w_1 = \frac{n(q-1)}{q}$ if and only if the defining set of $\mathcal{C}$ is the complement of a $t$-dimensional projective subspace in the projective space $\mathrm{PG}(k-1, q)$, where $t \leq k - 2$.
\end{theorem}

	\begin{proof}
	Since $\mathcal{C}$ is a projective two-weight $[n, k]$ code over $\mathbb{F}_q$, it follows that $n < \frac{q^k - 1}{q - 1}$. According to Proposition~\ref{prop:codimension-1}, there exist exactly $n$ codimension-$1$ subspaces $\mathcal{B}$ of $\mathcal{C}$ such that $\mathrm{wt}(\mathcal{B}) = n - 1$. Therefore, there exists at least one codimension-$1$ subspace $\mathcal{B}$ of $\mathcal{C}$ satisfying $\mathrm{wt}(\mathcal{B}) = n$.
From Proposition~\ref{prop:codimension-1}, we have
	\[
	A_{w_2} \left ( \mathcal{B} \right ) = \frac{1}{w_1 - w_2} \left( w_1 \left( q^{k-1} - 1 \right) - n q^{k-2} (q - 1) \right).
	\]
	Clearly, we also have $A_{w_2} \left ( \mathcal{B} \right ) \le A_{w_2} \left ( \mathcal{C} \right )$. Then, applying Proposition~\ref{prop:wt-dist-2}, it follows that
	\[
	w_1 \leq \frac{n(q - 1)}{q}.
	\]

	To derive a lower bound on $w_2$, we choose a codimension-$1$ subcode $\mathcal{B}$ of $\mathcal{C}$ such that $\mathrm{wt}(\mathcal{B}) = n - 1$. By Proposition~\ref{prop:codimension-1}, we have
\[
A_{w_1}(\mathcal{B}) = \frac{1}{w_2 - w_1} \left( w_2 \left( q^{k-1} - 1 \right) - (n-1) q^{k-2} (q - 1) \right).
\]
Using the inequality \(A_{w_1}(\mathcal{B}) \le A_{w_1}(\mathcal{C})\) together with Proposition~\ref{prop:wt-dist-2}, we obtain
\[
q w_2 \ge n(q - 1) + 1.
\]
We now prove by contradiction that equality cannot hold. Suppose, for contradiction, that
\[
q w_2 = n(q - 1) + 1.
\]
Let \(\mathcal{C}'\) be the linear span of all codewords in \(\mathcal{C}\) of Hamming weight \(w_1\), i.e., \(\mathcal{C}' = \mathrm{Span}(\mathrm{wt}^{-1}(w_1))\), and denote its dimension by \(k_1\). By Proposition~\ref{prop:codimension-1}, we know that for any codimension-1 subspace \(\mathcal{B} \subset \mathcal{C}\) with generalized weight \(n-1\), we have \(A_{w_1}(\mathcal{B}) = A_{w_1}(\mathcal{C})\). This implies \(k_1 < k\). 
Again, by Proposition~\ref{prop:codimension-1}, the set of codewords \(\mathrm{wt}^{-1}(w_1)\) forms an equidistant linear code. Applying Lemma~\ref{lemma:equidistant-support}, we conclude that
\[
\mathrm{wt}^{-1}(w_1) = \mathcal{C}' \setminus \{0\}.
\]
Hence, the total number of weight-\(w_1\) codewords in \(\mathcal{C}\) is
\[
A_{w_1}(\mathcal{C}) = \frac{1}{w_2 - w_1} \left( w_2(q^k - 1) - nq^{k-1}(q - 1) \right) = q^{k_1} - 1.
\]
We obtain:
\begin{equation} \label{eq:w1w2-relation}
(q^{k_1} - 1) w_1 = q^{k_1} w_2 - q^{k-1}.
\end{equation}
Now, by Proposition~\ref{prop:codimension-1}, the set \(\mathrm{wt}^{-1}(w_1)\) is covered by exactly \(n\) codimension-1 subspaces of \(\mathcal{C}\). On the other hand, the number of codimension-1 subspaces in \(\mathcal{C}\) that cover \(\mathcal{C}'\) is given by
\[
n = \frac{q^{k - k_1} - 1}{q - 1}.
\]
Under our assumption \(q w_2 = n(q - 1) + 1\), we substitute to get
\[
w_2 = q^{k - k_1 - 1}.
\]
Plugging this into equation~\eqref{eq:w1w2-relation}, we obtain:
\[
(q^{k_1} - 1) w_1 = q^{k_1} \cdot q^{k - k_1 - 1} - q^{k-1} = q^{k - 1} - q^{k - 1} = 0,
\]
which implies \(w_1 = 0\). This contradicts the assumption that \(w_1 > 0\) since it is a Hamming weight.
Therefore, the assumption \(q w_2 = n(q - 1) + 1\) cannot hold, and we must have
\[
q w_2 \ge n(q - 1) + 2,
\]
which gives the lower bound
\[
w_2 \ge \frac{n(q - 1) + 2}{q}.
\]

We now consider the characterization of two-weight codes with $w_1 = \frac{n(q - 1)}{q}$. 

First, suppose that the code $\mathcal{C}$ is defined by the complement of a $t$-dimensional projective subspace in $\mathrm{PG}(k-1, q)$, where $t \leq k - 2$. Then, it is straightforward to verify that $\mathcal{C}$ is a two-weight linear code of parameters 
\[
\left[ \frac{q^k - q^{t+1}}{q - 1},\; k \right]
\]
with weights 
\[
w_1 = q^{k-1} - q^t,\quad w_2 = q^{k-1}.
\]
Clearly, we have \( w_1 = \frac{n(q - 1)}{q} \).

Conversely, suppose that $\mathcal{C}$ is an $[n,k]$ two-weight code with $w_1 = \frac{n(q - 1)}{q}$. Let $\mathcal{C}'$ denote the linear span of all codewords in $\mathcal{C}$ of weight $w_2$, i.e., $\mathcal{C}' = \mathrm{Span}(\mathrm{wt}^{-1}(w_2))$, and denote its dimension by $k_2$.
From the proof of the inequality \( w_1 \leq \frac{n(q - 1)}{q} \), we know that for every codimension-1 subcode $\mathcal{B} \subset \mathcal{C}$ with full support (i.e., $\mathrm{wt}(\mathcal{B}) = n$), it holds that
\[
A_{w_2}(\mathcal{B}) = A_{w_2}(\mathcal{C}).
\]
This implies that $\mathcal{C}'$ is a proper subspace of $\mathcal{C}$, i.e., $k_2 < k$.
Applying Proposition~\ref{prop:codimension-1}, we know that the code with defining set \(\mathrm{wt}^{-1}(w_2)\) is a constant-weight linear code. By Lemma~\ref{lemma:equidistant-support}, it follows that:
\[
\mathrm{wt}^{-1}(w_2) = \mathcal{C}' \setminus \{0\}.
\]
On one hand, Proposition~\ref{prop:codimension-1} implies that exactly $\frac{q^k - 1}{q - 1} - n$ codimension-1 subspaces of $\mathcal{C}$ cover $\mathrm{wt}^{-1}(w_2)$. On the other hand, the number of codimension-1 subspaces in $\mathcal{C}$ that contain $\mathcal{C}'$ is:
\[
\frac{q^{k - k_2} - 1}{q - 1}.
\]
Equating the two expressions gives:
\[
\frac{q^k - 1}{q - 1} - n = \frac{q^{k - k_2} - 1}{q - 1},
\]
from which we obtain:
\[
n = \frac{q^k - q^{k - k_2}}{q - 1}.
\]
Substituting into \( w_1 = \frac{n(q - 1)}{q} \), we get:
\[
w_1 = q^{k - 1} - q^{k - k_2 - 1}.
\]
Moreover, since
\[
A_{w_2}(\mathcal{C}) = \# \mathrm{wt}^{-1}(w_2) = q^{k_2} - 1,
\]
we apply Proposition~\ref{prop:wt-dist-2} to conclude:
\[
w_2 = q^{k - 1}.
\]
Let the defining set of $\mathcal{C}$ be $D = \{\mathbf{v}_1, \dots, \mathbf{v}_n\} \subset \mathrm{PG}(k-1, q)$. Let the complement of $D$ in the projective space be denoted by $\{\mathbf{v}_1', \dots, \mathbf{v}_m'\}$, where \( m = \frac{q^k - 1}{q - 1} - n \).
Define the set
\[
W_2 := \left\{ \mathbf{u} \in \mathbb{F}_q^k : \mathrm{wt} \left( \langle \mathbf{u}, \mathbf{v}_i \rangle_{1 \le i \le n} \right) = w_2 \right\}.
\]
Then $W_2$ is a linear space isomorphic to $\mathrm{wt}^{-1}(w_2)$, and for every $\mathbf{u} \in W_2$ and $1 \le j \le m$, we have $\langle \mathbf{u}, \mathbf{v}_j' \rangle = 0$. Hence, all vectors in $W_2$ are orthogonal to all points in $\mathrm{PG}(k - 1, q) \setminus D$.
This implies:
\[
(q - 1) \cdot \left( \frac{q^k - 1}{q - 1} - n \right) \leq q^{k - k_2} - 1,
\]
which gives the inequality:
\[
n \ge \frac{q^k - q^{k - k_2}}{q - 1}.
\]
In fact, equality holds if and only if $\mathrm{PG}(k - 1, q) \setminus D$ coincides with the projective space associated with the orthogonal complement $W_2^{\perp}$. That is:
\[
D = \mathrm{PG}(k - 1, q) \setminus \overline{W_2^{\perp}} ,
\]
where $\overline{W_2^{\perp}}$ is the projective point set corresponding to $W_2^{\perp}$.
This completes the proof.

\end{proof}

Theorem~\ref{thm:w1-w2-bound} yields further structural insights as summarized below.

\begin{corollary}\label{cor:comp-subspace-*}
Let $\mathcal{C}$ be a projective two-weight $[n, k]$ code over $\mathbb{F}_q$ with weights $w_1 = \frac{n(q-1)}{q} < w_2$. Define 
\[
k_2 = k - \log_q\left( q^k - n(q-1) \right).
\]
Then the defining set of $\mathcal{C}$ is the complement of a $(k - k_2 - 1)$-dimensional projective subspace $\overline{ U}$ in $\mathrm{PG}(k - 1, q)$.
Moreover, the projectivization $\overline{\mathcal{C}^*(w_2)}$ of ${\mathcal{C}^*(w_2)}$ is the simplex code with parameters $\left[\frac{q^{k_2} - 1}{q - 1},\ k_2\right]$.
In addition, the projectivization $\overline{\mathcal{C}^*(w_1)}$ of ${\mathcal{C}^*(w_1)}$ is the linear code whose defining set is the complement of the $(k_2 - 1)$-dimensional projective subspace $\overline{U^{\perp}}$ in $\mathrm{PG}(k - 1, q)$, where $U$ denotes the linear subspace obtained by lifting the projective subspace $\overline{U}$.

\end{corollary}

\begin{remark}
The following two families of projective two-weight codes attain the lower bound 
$\tfrac{n(q-1)+2}{q}$ for the maximum nonzero weight $w_2$ given in Theorem~\ref{thm:w1-w2-bound}.
\begin{enumerate}
    \item[(1)] Let $\mathcal{C}$ be the code obtained from the $[2^k-1,k]$ simplex code over $\mathbb{F}_2$ by puncturing one coordinate. 
    Then $\mathcal{C}$ is a projective two-weight $[2^k-2,k]$ code with weights 
    $w_1 = 2^{k-1}-1$ and $w_2 = 2^{k-1}$.
    
    \item[(2)] Let $D_1$ and $D_2$ be two disjoint $(\ell-1)$-dimensional subspaces of 
    $\mathrm{PG}(2\ell-1,q)$. Then the code $\mathcal{C}_{D_1 \cup D_2}$ defined by the set 
    $D_1 \cup D_2$ is a projective two-weight 
    $\left [\tfrac{2(q^\ell-1)}{q-1}, 2\ell \right ]$ code with weights 
    $w_1 = q^{\ell-1}$ and $w_2 = 2q^{\ell-1}$.
\end{enumerate}
\end{remark}

This naturally leads to the following open problem, which appears to be 
worthy of further investigation.

\begin{problem}
Completely characterize projective two-weight codes with weights 
$w_1 < w_2 = \tfrac{n(q-1)+2}{q}$.
\end{problem}

\subsection{Codes Associated with Fixed-Weight Codewords of Two-Weight Codes}

For a projective two-weight code $\C$ with weights $w_1$ and $w_2$, the parameters of $\mathcal{C}^*(w_1)$ and $\mathcal{C}^*(w_2)$ are given below.

\begin{theorem}\label{thm:dual-2wt}
Let $\mathcal{C}$ be a projective two-weight $[n, k]$ code over $\mathbb{F}_q$ with weights $w_1 < w_2$, and suppose $w_1 \neq \frac{n(q-1)}{q}$. Let $A_{w_1}$ and $A_{w_2}$ be as described in Proposition~\ref{prop:wt-dist-2}. Then:
\begin{enumerate}
    \item[(1)] $\mathcal{C}^*(w_1)$ is a two-weight $[A_{w_1}, k, w_{1,1}^*]$ code with weights
    \[
    w_{1,1}^* = \frac{q^{k-2}(q-1)}{w_2 - w_1} \left( w_2 q - n(q-1) - 1 \right), \quad
    w_{1,2}^* = \frac{q^{k-2}(q-1)}{w_2 - w_1} \left( w_2 q - n(q-1) \right).
    \]
    Its weight distribution is
    \[
    A_{w_{1,1}^*} \left( \mathcal{C}^*(w_1) \right) = n(q-1), \quad
    A_{w_{1,2}^*} \left( \mathcal{C}^*(w_1) \right) = q^k - 1 - n(q-1).
    \]
    
    \item[(2)] $\mathcal{C}^*(w_2)$ is a two-weight $[A_{w_2}, k, w_{2,1}^*]$ code with weights
    \[
    w_{2,1}^* = \frac{q^{k-2}(q-1)}{w_1 - w_2} \left( w_1 q - n(q-1) \right), \quad
    w_{2,2}^* = \frac{q^{k-2}(q-1)}{w_1 - w_2} \left( w_1 q - n(q-1) - 1 \right).
    \]
    Its weight distribution is
    \[
    A_{w_{2,1}^*} \left( \mathcal{C}^*(w_2) \right) = q^k - 1 - n(q-1), \quad
    A_{w_{2,2}^*} \left( \mathcal{C}^*(w_2) \right) = n(q-1).
    \]
\end{enumerate}
\end{theorem}

\begin{proof}
We first show that the set of codewords of $\mathcal{C}$ with weight $w_1$, denoted by $\mathrm{wt}^{-1}(w_1)$, spans the entire code $\mathcal{C}$. Suppose, for the sake of contradiction, that the subspace spanned by the weight $w_1$ codewords is contained in a subcode $\mathcal{B}\subset\mathcal{C}$ of codimension $1$.
 Then $\mathcal{B}$ would contain all the codewords of weight $w_1$, that is,
\[
A_{w_1}(\mathcal{B}) = A_{w_1}(\mathcal{C}).
\]
Intuitively, $\mathcal{B}$ must then contain as many weight $w_1$ codewords as possible. By Proposition~\ref{prop:codimension-1}, this forces $\mathcal{B}$ to have generalized Hamming weight $\mathrm{wt}(\mathcal{B}) = n-1$. Combining Proposition~\ref{prop:codimension-1} with Proposition~\ref{prop:wt-dist-2}, one obtains
\[
w_2 = \frac{n(q-1)+1}{q},
\]
which contradicts Theorem~\ref{thm:w1-w2-bound}, where it is shown that $w_2 \ge \frac{n(q-1)+1}{q}$. Hence, the set of codewords of weight $w_1$ indeed spans the whole code $\mathcal{C}$. In particular, by Proposition~\ref{prop:dim(C*)}, the dimension of $\mathcal{C}^*(w_1)$ is equal to $k$. Finally, by combining Corollary~\ref{cor:C*-distribution} with Proposition~\ref{prop:codimension-1}, one obtains the explicit weight distribution of $\mathcal{C}^*(w_1)$ given in the theorem. 

The argument for $\mathcal{C}^*(w_2)$ is entirely analogous, and we omit the details. This completes the proof.
\end{proof}

As a corollary of Theorem \ref{thm:dual-2wt}, we have the following.
\begin{corollary}\label{cor:proj-dual-2wt}
Let $\mathcal{C}$ be a projective two-weight $[n, k]$ code over $\gf_q$ with weights $w_1$ and $w_2$, where $w_1 < w_2$, $w_1 \neq \frac{n(q-1)}{q}$. Let $A_{w_1}$ and $A_{w_2}$ be as described in Proposition~\ref{prop:wt-dist-2}. Then
\begin{enumerate}
    \item[(1)] The projectivization $\overline{\mathcal{C}^*({w_1})}$ of ${\mathcal{C}^*({w_1})}$  is a projective two-weight $[A_{w_1}/(q-1), k, \overline{w}_{1,1}^*]$ code with weights $\overline{w}_{1,1}^*$ and $\overline{w}_{1,2}^*$, where 
    \[
    \overline{w}_{1,1}^* = \frac{q^{k-2}}{w_2 - w_1} \left( w_2 q - n (q-1) - 1 \right), \quad \text{and} \quad \overline{w}_{1,2}^* = \frac{q^{k-2}}{w_2 - w_1} \left( w_2 q - n (q-1) \right).
    \]
    Its weight distribution is given by
    \[
    A_{\overline{w}_{1,1}^*} \left (\overline{\mathcal{C}^*({w_1})} \right ) = n(q-1),  
    \quad     \text{and} \quad 
    A_{\overline{w}_{1,2}^*}\left (\overline{\mathcal{C}^*({w_1})} \right ) = q^k-1-n(q-1).
    \]

    \item[(2)] The projectivization $\overline{\mathcal{C}^*({w_2})}$ of ${\mathcal{C}^*({w_2})}$  is a projective two-weight $[A_{w_2}/(q-1), k, \overline{w}_{2,1}^*]$ code with weights $\overline{w}_{2,1}^*$ and $\overline{w}_{2,2}^*$, where 
    \[
    \overline{w}_{2,1}^* = \frac{q^{k-2}}{w_1 - w_2} \left( w_1 q - n (q-1) \right), \quad \text{and} \quad \overline{w}_{2,2}^* = \frac{q^{k-2}}{w_1 - w_2} \left( w_1 q - n (q-1) - 1 \right).
    \]
    Its weight distribution is given by
    \[
    A_{\overline{w}_{2,1}^*} \left ( \overline{\mathcal{C}^*({w_2})} \right ) = q^k-1-n(q-1), \quad 
    \text{and}
   \quad 
    A_{\overline{w}_{2,2}^*}\left (\overline{\mathcal{C}^*({w_2})} \right ) = n(q-1).
    \]
\end{enumerate}
\end{corollary}

The following result provides a more compact representation of the quantities 
$\overline{w}_{i,j}^*$ for $1 \le i,j \le 2$.

\begin{proposition}\label{prop:wij-wAw}
Let $\mathcal{C}$ be a projective two-weight $[n,k]$ code over $\mathbb{F}_q$ with nonzero weights $w_1$ and $w_2$ satisfying $w_1 < w_2$ and 
$w_1 \neq \tfrac{n(q-1)}{q}$. 
Suppose that the corresponding frequencies are $A_{w_1}$ and $A_{w_2}$, respectively. 
Let $\overline{w}_{1,1}^*$, $\overline{w}_{1,2}^*$, $\overline{w}_{2,1}^*$, and $\overline{w}_{2,2}^*$ be as defined in Corollary~\ref{cor:proj-dual-2wt}. 
Then we have
\[
\overline{w}_{1,1}^* = \frac{w_1 A_{w_1}}{n(q-1)}, 
\qquad 
\overline{w}_{2,2}^* = \frac{w_2 A_{w_2}}{n(q-1)},
\]
and
\[
\overline{w}_{1,2}^* = \frac{w_2' A_{w_1}}{n'(q-1)}, 
\qquad 
\overline{w}_{2,1}^* = \frac{w_1' A_{w_2}}{n'(q-1)},
\]
where 
$
n' = \frac{q^k-1}{q-1} - n$ and
$w_i' = q^{\,k-1} - w_{3-i}~ (i=1,2).
$
\end{proposition}

\begin{proof}
From the second Pless power moment identity, we have
\[
w_1 A_{w_1} + w_2 A_{w_2} = n q^{k-1}(q-1).
\]
Multiplying both sides by $w_2$ gives
\[
w_1 w_2 A_{w_1} + w_2^2 A_{w_2} = n w_2 q^{k-1}(q-1).
\]
On the other hand, the third Pless power moment yields
\[
w_1^2 A_{w_1} + w_2^2 A_{w_2} = n q^{k-2}(q-1)(qn-n+1).
\]
Subtracting the two equations above and simplifying, we obtain
\begin{equation}\label{eq:wAw-n}
\frac{w_1 A_{w_1}}{n(q-1)} = \frac{q^{k-2}}{w_2-w_1} \Big( q w_2 - n(q-1) - 1 \Big).
\end{equation}
By Part (1) of Corollary~\ref{cor:proj-dual-2wt}, this immediately implies that 
$\overline{w}_{1,1}^* = \tfrac{w_1 A_{w_1}}{n(q-1)}$. 
By symmetry between $w_1$ and $w_2$, it follows analogously that 
$\overline{w}_{2,2}^* = \tfrac{w_2 A_{w_2}}{n(q-1)}$.

Now let $\mathcal{C}'$ be the code corresponding to the complement of the set 
$\{ \mathbf{g}_1, \dots, \mathbf{g}_n \}$ in the projective space $\mathrm{PG}(k-1,q)$, 
where $\mathbf{g}_1, \dots, \mathbf{g}_n$ are the column vectors of a generator matrix of $\mathcal{C}$.  
Since $w_1 \neq \tfrac{n(q-1)}{q}$, by Theorem~\ref{thm:w1-w2-bound} we know that $\mathcal{C}'$ is a projective two-weight $[n',k]$ code with weights $w_1'<w_2'$ and with corresponding frequencies $A_{w_2}$ and $A_{w_1}$, respectively. 

For the code $\mathcal{C}'$, an equation analogous to \eqref{eq:wAw-n} holds:
\[
\frac{w_1' A_{w_1'}(\mathcal{C}')}{n'(q-1)} 
= \frac{q^{k-2}}{w_2' - w_1'} \Big( q w_2' - n'(q-1) - 1 \Big).
\]
Substituting $A_{w_1'}(\mathcal{C}') = A_{w_2}$ on the left-hand side and 
$w_i' = q^{k-1} - w_{3-i}$ together with $n' = \tfrac{q^k - 1}{q-1} - n$ on the right-hand side, we obtain
\[
\overline{w}_{2,1}^* = \frac{w_1' A_{w_2}}{n'(q-1)}.
\]
Finally, by symmetry between $w_1$ and $w_2$, we also deduce
\[
\overline{w}_{1,2}^* = \frac{w_2' A_{w_1}}{n'(q-1)}.
\]
This completes the proof.
\end{proof}

\begin{example}
Let $\alpha$ be the primitive element of $\mathbb{F}_{3^6}$ such that $\alpha^6 -\alpha^4 + \alpha^2 -\alpha + 2=0$. Define 
\[
D = \left\{ \alpha^i : \mathrm{Tr}^6_1 \left (\alpha^{2i} \right ) = 0 \ \text{and} \ 0 \leq i \leq \frac{3^6 - 1}{3 - 1} - 1 \right\}.
\]
Then the linear code $\mathcal{C}_D$ has parameters $[112, 6 ,72]$ and weight enumerator $x^{112} +  504 x^{40}y^{72}+ 224 x^{30}y^{81}$. The linear code $\overline{\mathcal{C}_D^*(72)}$ has parameters $[ 252, 6,  162]$ and weight enumerator $x^{252} + 224 x^{90}y^{162} + 504x^{81}y^{171}$. The linear code $\overline{\mathcal{C}_D^*(81)}$ has parameters $[112, 6 ,72]$ and weight enumerator $x^{112} +  504 x^{40}y^{72}+ 224 x^{30}y^{81}$.

\end{example}

The following result shows that the parameters of a projective two-weight code satisfy certain divisibility conditions.

\begin{theorem}
Let $\mathcal{C}$ be a projective two-weight $[n,k]$ code over $\mathbb{F}_q$ with nonzero weights 
$w_1 < w_2$ and corresponding frequencies $A_{w_1}$ and $A_{w_2}$. Then the following hold:
\begin{enumerate}
    \item[(1)] The weight difference $w_2 - w_1$ divides each of $q^{\,k-2}$, $w_1$, and $w_2$;
    \item[(2)] The length $n$ divides both $\frac{w_1 A_{w_1}}{q-1}$ and $\frac{w_2 A_{w_2}}{q-1}$.
\end{enumerate}
\end{theorem}

\begin{proof}
From Part~(1) of Corollary~\ref{cor:proj-dual-2wt}, we have
\[
(w_2 - w_1)\big(\overline{w}_{1,2}^* - \overline{w}_{1,1}^*\big) = q^{\,k-2}.
\]
Hence, $w_2 - w_1$ divides $q^{\,k-2}$. Furthermore, by using the explicit expressions for 
$A_{w_1}$ and $A_{w_2}$ given in Proposition~\ref{prop:wt-dist-2}, it follows immediately that both $w_1$ and $w_2$ are divisible by $w_2 - w_1$.
Part~(2) of the theorem follows directly from Corollary~\ref{cor:proj-dual-2wt} and Proposition~\ref{prop:wij-wAw}.
\end{proof}

\begin{theorem}\label{thm:2weight-C-C**}

Let $\mathcal{C}$ be a projective two-weight $[n, k]$ code over $\mathbb{F}_q$ with nonzero weights $w_1$ and $w_2$, where $w_1 < w_2$ and $w_1 \neq \tfrac{n(q-1)}{q}$. Let $\mathbf{g}_1, \dots, \mathbf{g}_n$ denote the column vectors of a generator matrix of $\mathcal{C}$. Let the notations $\overline{\mathcal{C}^*(w_1)}$, $\overline{\mathcal{C}^*(w_2)}$, and the parameters $\overline{w}_{1,1}^*$, $\overline{w}_{1,2}^*$, $\overline{w}_{2,1}^*$, and $\overline{w}_{2,2}^*$ be as defined in Corollary~\ref{cor:proj-dual-2wt}. Then we have
\[
\overline{ \overline{\mathcal{C}^*(w_1)}^*(\overline{w}_{1,1}^*) } \cong \overline{ \overline{\mathcal{C}^*(w_2)}^*(\overline{w}_{2,2}^*) } \cong \mathcal{C},
\]
and
\[
\overline{ \overline{\mathcal{C}^*(w_1)}^*(\overline{w}_{1,2}^*) } \cong \overline{ \overline{\mathcal{C}^*(w_2)}^*(\overline{w}_{2,1}^*) } \cong \mathrm{Comp}(\mathcal{C}),
\]
where $\mathrm{Comp}(\mathcal{C})$ denotes the code corresponding to the complement of the set $\{ \mathbf{g}_1, \dots, \mathbf{g}_n \}$ in the projective space $\mathrm{PG}(k-1,q)$, and the symbol \(\cong\) stands for code equivalence.
	
\end{theorem}

\begin{proof}
By Corollary~\ref{cor:proj-dual-2wt}, since $w_1 \neq \tfrac{n(q-1)}{q}$, the codewords of $\mathcal{C}$ of weight $w_1$ span the entire code $\mathcal{C}$. Combining Corollary~\ref{cor:proj-dual-2wt} with Proposition~\ref{prop:wij-wAw} yields 
\[
\overline{w}_{1,1}^*=\frac{w_1 A_{w_1}}{n(q-1)} 
\quad \text{and} \quad 
A_{\overline{w}_{1,1}^*} \big( \overline{\mathcal{C}^*(w_1)} \big) = n(q-1).
\]
Since $\mathcal{C}$ is a projective two-weight code, by Theorem \ref{thm:AGM}, the codewords of weight $w_1$ support a $1$-design. Thus, Theorem~\ref{thm:C**-C} implies 
\[
\overline{ \overline{\mathcal{C}^*(w_1)}^*(\overline{w}_{1,1}^*) } \cong \mathcal{C}.
\]
By symmetry, we also have 
\[
\overline{ \overline{\mathcal{C}^*(w_2)}^*(\overline{w}_{2,2}^*) } \cong \mathcal{C}.
\]

Next, observe that $\mathrm{Comp}(\mathcal{C})$ is itself a projective two-weight $[\tfrac{q^k-1}{q-1}-n, k]$ code with nonzero weights 
\[
w_1' = q^{k-1} - w_2, 
\quad 
w_2' = q^{k-1} - w_1,
\]
and corresponding frequencies $A_{w_2}$ and $A_{w_1}$. It is straightforward to see that the set of codewords of weight $w_1'$ in $\mathrm{Comp}(\mathcal{C})$ is linearly equivalent to the set of codewords of weight $w_2$ in $\mathcal{C}$. Hence,
\[
\overline{\mathrm{Comp}(\mathcal{C})^*(w_1')} \cong \overline{\mathcal{C}^*(w_2)}.
\]
Again, combining Corollary~\ref{cor:proj-dual-2wt} and Proposition~\ref{prop:wij-wAw} gives 
\[
\overline{w}_{2,1}^* = \frac{w_1' A_{w_2}}{n'(q-1)} 
\quad \text{and} \quad 
A_{\overline{w}_{2,1}^*}\big( \overline{\mathrm{Comp}(\mathcal{C})^*(w_1')} \big) = n'(q-1),
\]
where $n' = \tfrac{q^k-1}{q-1} - n$. Applying Theorem~\ref{thm:C**-C}, we obtain
\[
\overline{ \overline{\mathrm{Comp}(\mathcal{C})^*(w_1')}^*(\overline{w}_{2,1}^*) } \cong \mathrm{Comp}(\mathcal{C}).
\]
Consequently,
\[
\overline{\overline{\mathcal{C}^*(w_2)}^*(\overline{w}_{2,1}^*) } \cong \mathrm{Comp}(\mathcal{C}).
\]
By symmetry, the same reasoning shows that
\[
\overline{ \overline{\mathcal{C}^*(w_1)}^*(\overline{w}_{1,2}^*) } \cong \mathrm{Comp}(\mathcal{C}).
\]
This completes the proof.
\end{proof}

\begin{corollary}
Let $\mathcal{C}$ be a projective two-weight code with weights $w_1 < w_2$. Then the code $\mathcal{C}^*(w_1)$ has dimension $k$. Moreover, we have
\[
\overline{ \overline{\mathcal{C}^*(w_1)}^{*}(\overline{w}_{1,1}^*) } \cong \mathcal{C},
\]
where $\overline{w}_{1,1}^*$ denotes the minimum nonzero weight of the code $\overline{\mathcal{C}^*(w_1)}$.
\end{corollary}

\begin{proof}
The conclusion follows directly from Corollary~\ref{cor:comp-subspace-*} and Theorem~\ref{thm:2weight-C-C**}.
\end{proof}

\subsection{Characterization of Extendable Two-Weight Codes}

In \cite{SDC2024}, Sun, Ding, and Chen investigated the extended codes of linear codes, and in \cite{Sun2024}, they determined the weight distributions of the extended codes for several classes of two-weight codes. Nevertheless, the fundamental question of which two-weight codes are extendable remains open. The following theorem provides a characterization of extendable two-weight codes.

\begin{theorem}
Let $\mathcal{C}$ be an $[n,k]$ two-weight linear code over $\mathbb{F}_q$ with nonzero Hamming weights $w_1 < w_2$. Then $\mathcal{C}$ is extendable if and only if
\[
w_2 = \frac{(q - 1) q^{k - 2} n}{q^{k - 1} - 1}.
\]
\end{theorem}

\begin{proof}
We first consider the case where the defining set of $\mathcal{C}$ is the complement of a $(t-1)$-dimensional projective subspace $U$ in $\mathrm{PG}(k-1, q)$, where $1 \le t \le k - 1$. In this case, $\mathcal{C}$ is a two-weight code with parameters $\left[\frac{q^k - q^t}{q - 1},\, k\right]$ and weights $w_1 = q^{k-1} - q^{t-1}$ and $w_2 = q^{k-1}$. Let $\tilde{\mathcal{C}}$ denote the code defined by the point set $\left(\mathrm{PG}(k-1, q) \setminus U\right) \cup \{P\}$. Then $\tilde{\mathcal{C}}$ is an extension of $\mathcal{C}$. When $k - 1 \ge t \ge 2$, it can be directly verified that $\tilde{\mathcal{C}}$ and $\mathcal{C}$ share the same minimum nonzero distance, which implies that $\mathcal{C}$ is not extendable. Moreover, one can easily check that in this case $w_2 > \frac{(q - 1) q^{k - 2} n}{q^{k - 1} - 1}$, and hence the conclusion of the theorem holds. For $t = 1$, Example~\ref{remark:PG-POINT} shows that the conclusion still holds.

Now suppose that the defining set of $\mathcal{C}$ is not the complement of a projective subspace. By Corollary~\ref{cor:proj-dual-2wt}, the code $\overline{\mathcal{C}^*(w_1)}$ has dimension $k$. According to Theorem~\ref{thm:ext-dist}, the code $\mathcal{C}$ is extendable if and only if $\overline{\mathcal{C}^*(w_1)}$ contains a codeword whose weight equals its length. Again, from Corollary~\ref{cor:proj-dual-2wt}, the length of $\overline{\mathcal{C}^*(w_1)}$ is $A_{w_1}/(q - 1)$, and its largest nonzero weight is given by
\[
\overline{w}_{1,2}^* = \frac{q^{k - 2}}{w_2 - w_1} \left(w_2 q - n (q - 1)\right),
\]
where 
\[
A_{w_1} = \frac{1}{w_2 - w_1} \left(w_2 (q^k - 1) - n q^{k - 1} (q - 1)\right).
\]
Hence, $\mathcal{C}$ is extendable if and only if
\[
\frac{1}{(q - 1)(w_2 - w_1)} \left(w_2 (q^k - 1) - n q^{k - 1} (q - 1)\right) = \frac{q^{k - 2}}{w_2 - w_1} \left(w_2 q - n (q - 1)\right).
\]
A direct computation shows that this equality holds precisely when
\[
w_2 = \frac{(q - 1) q^{k - 2} n}{q^{k - 1} - 1}.
\]
The proof is thus complete.
\end{proof}

In what follows, we describe three families of extendable two-weight codes.

\begin{example}\label{remark:PG-POINT}
Let $\mathcal{C}$ be the linear code corresponding to the point set obtained by deleting one point from the projective space $\mathrm{PG}(k-1, q)$. Then $\mathcal{C}$ is a two-weight linear code with parameters $[n, k] = [\frac{q^k - q}{q - 1},\, k]$ and nonzero weights $w_1 = q^{k-1} - 1$ and $w_2 = q^{k-1}$. Consequently, these parameters satisfy 
\[
w_2 = \frac{(q - 1) q^{k-2} n}{q^{k-1} - 1}.
\]
It is evident that the $q$-ary $k$-dimensional simplex code, with parameters $[\frac{q^k - 1}{q - 1},\, k]$ and minimum weight $q^{k-1}$, can be viewed as an extension of $\mathcal{C}$. Therefore, $\mathcal{C}$ is extendable.
\end{example}

\begin{example}\label{remark:hyperoval}
A hyperoval in $\mathrm{PG}(2,q)$, where $q$ is even, is a set of $q+2$ points such that no three are collinear and every line intersects it in either $0$ or $2$ points. While the hyperoval is unique for $q=2$ and $q=4$, there exist many projectively distinct examples for larger even $q$. 
From the coding-theoretic viewpoint, the linear code $\mathcal{C}_D$ associated with a hyperoval $D$ is an MDS code with parameters $[q+2,3,q]$, which has exactly two nonzero weights, namely $w_1 = q$ and $w_2 = q+2$. Furthermore, the code $\overline{\mathcal{C}_D^*(w_1)}$ is a two-weight code with parameters $\left[\frac{(q+1)(q+2)}{2}, 3\right]$, where the nonzero weights are $w_1 = \frac{q(q+1)}{2}$ and $w_2 = \frac{q(q+2)}{2}$. Therefore, $\overline{\mathcal{C}_D^*(w_1)}$ is an extendable two-weight code.
\end{example}

\begin{example}\label{remark:max-arc}
Let $q$ and $h$ be powers of $2$ with $1 < h < q$ and $h \mid q$. 
Choose an irreducible quadratic form over $GF(q)$, say 
\(
\phi(x,y) = ax^2 + bxy + cy^2,
\)
and let $H$ be a subgroup of the additive group of $GF(q)$ of order $h$. 
Define a subset $D$ of the projective plane $\mathrm{PG}(2,q)$ by
\(
D = \{(1, x, y) : \phi(x, y) \in H\}. 
\)
Denniston~\cite{Denniston1969} proved that the linear code $\mathcal{C}_D$ associated with $D$ 
is a two-weight code with parameters $[(h-1)(q+1)+1, 3]$ and nonzero weights
\(
w_1 = (h-1)q, \quad w_2 = (h-1)(q+1)+1.
\)
Moreover, the code $\overline{\mathcal{C}_D^*(w_1)}$ has parameters
\(
\left [\frac{(q+1)(1+(h-1)(q+1))}{h}, \, 3 \right ]
\)
and two nonzero weights
\(
w_1 = \frac{q(q+1)(h-1)}{h}, \quad 
w_2 = \frac{q(1+(q+1)(h-1))}{h}.
\)
Hence, $\overline{\mathcal{C}_D^*(w_1)}$ is an extendable two-weight code.
\end{example}

Based on Examples~\ref{remark:PG-POINT}, \ref{remark:hyperoval} and \ref{remark:max-arc}, the following open problem is both interesting and worthy of further investigation.  

\begin{problem}
Does there exist an extendable two-weight linear code over $\mathbb{F}_q$ whose parameters differ from $\left[\frac{(q+1)\bigl(1+(h-1)(q+1)\bigr)}{h},\,3\right]$ or $\left[\frac{q^k - q}{q - 1},\,k\right]$?
\end{problem}

\section{Conclusion}\label{conc}
In this paper, we proposed a new framework for the secondary construction of linear codes based on their weight functions. By considering the set of codewords in a given linear code having a fixed Hamming weight, we developed a systematic and general approach for generating new linear codes and analyzing their fundamental parameters. This methodology establishes intrinsic connections between the structural properties of the constructed codes and the extendability of the original ones. For projective two-weight codes, we completely determined the parameters of the derived codes, established an upper bound on the minimum weight of two-weight codes, and characterized all two-weight codes attaining this bound. In addition, several divisibility properties concerning the parameters of two-weight codes were obtained, further revealing the algebraic regularities behind their weight functions.

Beyond the specific class of two-weight codes examined in this work, the proposed framework provides a versatile analytical tool for both code construction and structural exploration.  Future research may naturally extend this approach to other significant families of linear codes, such as BCH codes, algebraic-geometric codes, and LDPC codes, by investigating the linear codes associated with their fixed-weight codewords. We believe that such studies will not only yield new and interesting code families but also deepen our understanding of the combinatorial and algebraic foundations of coding theory.

\end{document}